\documentclass[aps,pra,twocolumn,10pt,nofootinbib]{revtex4-1}
\usepackage[utf8]{inputenc}
\usepackage{graphicx}
\usepackage{dcolumn}
\usepackage{bm}
\usepackage{multirow}
\usepackage{amssymb}
\usepackage{amsfonts}
\usepackage{amsmath}
\usepackage{mathtools}
\usepackage{simplewick}
\usepackage{amsthm}
\usepackage{enumerate}
\usepackage{color}
\usepackage{marvosym}
\newcommand{\mathbbm}[1]{\text{\usefont{U}{bbm}{m}{n}#1}}
\newtheorem{theorem}{Theorem}
\newtheorem{observation}{Observation}

\newcommand{\id}{\ensuremath{\mathbbm{1}}}
\newcommand{\one}{\id}

\newcommand{\bra}[1]{\langle #1|}
\newcommand{\ket}[1]{|#1\rangle}

\newcommand{\proj}[1]{| #1\rangle \langle #1|}

\newcommand{\be}{\begin{equation}}
\newcommand{\ee}{\end{equation}}
\newcommand{\bea}{\begin{eqnarray}}
\newcommand{\eea}{\end{eqnarray}}

\newcommand{\identity}{\mathbbm{1}}

\let\vec\mathbf 

\newcommand{\forget}[1]{}

\begin{document}

\title{Computational power of matchgates with supplementary resources}

\author{M. Hebenstreit$^1$, R. Jozsa$^2$, B. Kraus$^1$ and S. Strelchuk$^2$ }
\affiliation{$^1$Institute for Theoretical Physics, University of Innsbruck, Technikerstr. 21A, 6020 Innsbruck, Austria\\
$^2$DAMTP, University of Cambridge, Cambridge CB3 0WA, UK}

\begin{abstract}
We study the classical simulation complexity in both the weak and strong senses, of matchgate (MG) computations supplemented with all combinations of settings involving inclusion of intermediate adaptive or nonadaptive computational basis measurements, product state or magic and general entangled state inputs, and single- or multi-line outputs. We find a striking parallel to known results for Clifford circuits, after some rebranding of resources.
We also give bounds on the amount of classical simulation effort required in case of limited access intermediate measurements and entangled inputs. In further settings we show that adaptive MG circuits remain classically efficiently simulable if arbitrary two-qubit entangled input states on consecutive lines are allowed, but become quantum universal for three or more lines.
 And if adaptive measurements in non-computational bases are allowed, even with just computational basis inputs, we get quantum universal power again. 
\end{abstract}

\maketitle

\section{Introduction}

When quantum computation is implemented in an actual physical system there is often an associated  particular  restricted class of physically natural gates. Two examples are the quantum computational
processes realised in the physics of  bosonic resp. fermionic, linear optics. In the bosonic case e.g. in the KLM model \cite{klm}, we have the operations of beam splitter transformations, phase-shifting gates and (possibly adaptive) projective measurements. In the fermionic case the corresponding quantum processes are captured by the elegant mathematical formalism of matchgate (MG) circuits~\cite{Va02,TeDi02,JoMi08}. Another important restricted gate set is that of Clifford gates, featuring fundamentally in a variety of aspects of quantum computation and information. 
They are especially important in the context of quantum error correction and fault tolerant quantum computation \cite{NiCh00, GottesmanThesis}, and measurement-based quantum computation \cite{RaBr01} amongst others.
The notion of Clifford gate too, rests on an elegant mathematical (group theoretic) formalism but unlike matchgates, it appears not to have an associated characteristic physical basis.

In this work we will explore the computational power of MG circuits supplemented with some further resources, and along the way we will point out some interesting parallels and contrasts to properties of Clifford circuits.

In order to obtain insight into the features which are responsible for the power of quantum computation, a fruitful quantitative approach is to investigate the complexity of its classical simulation. Furthermore, delicate relations between classically efficiently simulable quantum computations and universal ones, as can be exposed in such investigations, provide a novel approach \cite{JoSt17}  to the important issue of verification of quantum computations, in the regime beyond the remit of efficient  classical computation.

MG circuits and Clifford circuits each have an especially rich associated theory of classical simulation properties ranging from being classically efficiently simulable to being fully quantum computationally universal after inclusion of restricted kinds of further resources. For Clifford circuits, classical simulation properties begin with the celebrated Gottesman-Knill theorem \cite{NiCh00, GottesmanThesis} and inclusion of a variety of further resources is studied in \cite{JoVa14,koh}. Classical simulation properties of MG circuits are rooted in the solvability of the free fermion model in many body physics \cite{TeDi02} and they were introduced into quantum computation with the work of Valiant \cite{Va02}. Below we will give a review of established techniques for simulation of MG circuits, as this will form an underpinning for our principal results.

Mathematically, MGs are a particular class of 2-qubit
gate defined by some algebraic equations (cf
Sec. \ref{sec:preliminaries} below for a precise definition). In particular they always preserve the even and odd parity subspaces of the 2-qubit state space. MG circuits are circuits in which MGs are allowed to act only on {\em nearest neighbour} qubits. In addition to rich classical simulation properties (elaborated below) they also possess a number of further striking properties of especial interest in quantum computing. For example, under certain conditions one can `compress' the MG computation to exponentially fewer qubits, offering an exponential space saving \cite{JoKr10}. The required nearest neighbour qubit action gives rise to further interesting features related to the topology of the arrangement of qubit lines: MG circuits on qubits arranged in a ring or a straight line are classically efficiently simulatable whereas any other topology yields circuits capable of universal quantum computation~\cite{BrGa12, BrCh14}. Universal quantum power is also obtained when in addition to MGs one is allowed use of any further non-MG parity-preserving 2-qubit gate~\cite{BrGa11}.

It is well known that the quantum computational power of any (generally restricted) circuit model depends on more ingredients than just specification of the allowed gate set, and consideration of such further ingredients will be key for our main results.
We have at our choice the specification of allowable input states, such as computational basis states or general product states or so-called magic states. For outputs we may allow only a single qubit measurement (corresponding to decision problems with just a yes/no answer) or multiple qubit measurements to sample more general distributions. We may also consider intermediate measurements within the body of the circuit, which can then be non-adaptive or adaptive i.e. having their classical outcomes possibly determining the choice of subsequent gates. All of these, and in their various combinations, generally have profound effects on classical simulation properties and thus on the computational power of the allowed resources. For example a change from non-adaptive to adaptive measurements even by itself, can elevate the computational power from classically efficiently simulable to quantum universal \cite{JoVa14}. The fertility of adaptive measurements is also strikingly reflected in measurement based quantum computation \cite{RaBr01}. In addition to such possible supplementary resources we also have a variety of notions of efficient classical simulation. Here we will use classical simulation in the strong sense and in the weak sense. Roughly speaking, a polynomial time quantum computation on $n$ qubits is called efficiently classically simulable in the strong sense if each output probability and marginal can be calculated classically up to a precision of $m$ digits in classical poly$(n,m)$ time. Better suited for a fair comparison to the actual running of a quantum computer is the notion of efficient weak classical simulation, wherein we ask to be able to sample the output probability distribution of the quantum computation in randomised classical polynomial time (see Sec. \ref{sec:preliminaries}).

It is known (cf \cite{JoVa14,TeDi02,JoMi08})  that any circuit composed of Clifford gates or of matchgates is classically simulable in the strong sense, if the input is a computational basis state and the output is a final measurement on a single qubit. Deviating slightly from this setting can lead to computations which are still classically efficiently simulable, however, possibly only in the {\it weak} sense. Deviating even further, leads to computations which are very unlikely to be classically efficiently simulable. Similarly, bosonic linear optics can be classically efficiently simulated under certain constraints, but allowing adaptive measurements leads to quantum universal computing power \cite{klm}.

If we allow the input state to be any product state (still with single line output), the simulation complexity of Clifford circuits varies from classically efficiently simulable  to \#P-hard (which includes NP-hardness), depending on whether adaptive measurements are allowed or not. In contrast, the simulation of MG circuits on product states remains classically efficiently simulable~\cite{JoMi08} even with adaptive measurements being allowed~\cite{Br16}. This still holds true if in addition to MG circuits one is allowed to perform an arbitrary single-qubit gate at any stage on the first qubit line~\cite{JoMi15,Br16}. Recently, the scope of efficient (weak) classical simulability of MG circuits has been extended to the scenario of having an arbitrary product state input, arbitrary (many-line) output measurements in arbitrary 1-qubit bases, and adaptive measurements in the computational basis all being simultaneously included~\cite{Br16}. In a more experimentally realistic setting, fermionic linear optical processes retain classical simulability properties also in the presence of suitable noisy ancillas \cite{OsGu14}.

The strikingly different classical simulation complexity properties of Clifford circuits vs. MG circuits, particularly in a scenario of product state inputs and adaptive measurements, are also reflected in the notion of magic state as it applies to these two classes of circuits \cite{BrKi05,HeJo19}.
In essence, magic state inputs together with the availability of adaptive measurements provide a means for elevating the computational power of a class of circuits, here Clifford circuits or MG circuits, to universal quantum computing power. This is achieved by so-called gate gadget constructions, utilising the resource of magic state inputs and adaptive measurements to implement a further new gate giving overall a universal gate set. In the Clifford case the magic states can be single qubit states, which can be chosen to deterministically implement the $T$--gate \cite{BrKi05}. In contrast, for MG circuits the simplest pure magic state input is a 4--qubit entangled state \cite{HeJo19}, while arbitrary 1-qubit state inputs (i.e. product states) with adaptive measurements remain all classically efficiently simulable.

Hence, the number of qubits which can be entangled in the input state has a marked effect on the simulation complexity in case of adaptive MG circuits (in contrast to Cliffords where product states already provide quantum universal power). Developing this realisation we will see below, as part of our main results,  that
an interesting comparison between the simulation complexities of Clifford circuits and MG circuits arises if one replaces the cases of inputs states being either computational basis states or product states (as in the Clifford circuit results of \cite{JoVa14}), by the cases of the inputs being either product states or magic states respectively. These results are summarised in Fig. \ref{fig:mg_magic_complexity}  that applies to MGs yet turns out to exactly reproduce the pattern of simulation complexities for Clifford circuits as given in \cite{JoVa14}.

This is a principal aim of the present work  viz. to determine the classical simulation cost of MG computations when supplemented with magic states and/or adaptive measurements, for both the single--line output and the multi--line output cases. Furthermore, we generalise these results to the scenarios where only a limited amount of either of these resources is available.
We will also address two natural associated issues, the simulation complexity when one has access to adaptive measurements in bases different from the computational basis, and when one is able to supply arbitrary entangled input states.

The paper is structured as follows. First we summarise our findings in Sec \ref{sec:Results}. In Section \ref{sec:preliminaries} we introduce our notation, present a short review of the two most commonly used techniques for classically simulating MG circuits, and we review the notion of magic states in the context of MG circuits. In Section~\ref{sec:clsim} we begin the development of our results, establishing first the classical simulation complexity of MG circuits in the setting allowing magic input states and adaptive measurements in the computational basis, and we consider its dependence on the amount of each of these resources that is used.
In Section~\ref{sec:hardnessclassification} we derive classical simulation complexities for further MG settings, treating various kinds of input state, type of measurements, the number of lines measured, and the type of simulation required. This will complete the demonstration of all the results shown in Fig. \ref{fig:mg_magic_complexity}, and establish the resemblance to corresponding results for Clifford gates claimed above. Finally, in Section \ref{sec:additionalSimu} we present our results on classical simulability of MG computations with two qualitatively different further types of resources. Firstly, in~\ref{sec:twoqubitinput} we consider MG circuits supplied with entangled input states on consecutive lines. We show that for the case of 2-qubit states, such circuits with adaptive measurements in the computational basis  can be classically efficiently simulated, while the case of 3- or more qubit states can provide universal quantum power. Secondly,
in Section~\ref{sec:adaptivemeas} we investigate MG circuits with the additional resource of adaptive measurements in arbitrary bases and show that this suffices again for universal quantum computing power.

\section{Summary of Results}
\label{sec:Results}

To be able to concisely refer to the broad variety of computational resources that we will be considering, we introduce associated abbreviations following those in \cite{JoVa14}.
 IN(BITS) and IN(PROD) will refer respectively to the setting where the input state is an arbitrary computational basis state or arbitrary product state. IN(MAGIC) will refer to the setting where the input state may contain also magic states in addition to product states. We will refer to a scenario having non-adaptive or adaptive intermediate measurements as ADAPT or NONADAPT respectively. Scenarios where we measure either a single or an arbitrary number of lines (possibly all lines) at the end of the circuit will be referred to as OUT(1) and OUT(MANY). Unless stated otherwise, we will always consider the final measurements to be measurements in the computational basis. We will use two standard types of classical simulation, weak and strong simulation denoted WEAK and STRONG, as defined in \cite{JoVa14} and below in Section \ref{sec:preliminaries}.
 Briefly, given a description of the computation, in the former case we wish to classically sample from the output distribution of the circuit, whereas in the latter case we wish to classically compute the probability of any designated outcome or marginal. In addition, we introduce abbreviations for various degrees of complexity of classical simulation of a scenario or setting being considered: Cl-P (``classical poly time") will  indicate that in the scenario can be {\em efficiently} classically simulated (in the weak or strong sense being considered). The  label QC-hard indicates that classical simulation of the scenario suffices for classical simulation of universal quantum computation (in the weak sense). Finally, \#P-hard  denotes that a classical simulation algorithm in this setting would also be capable of solving any problem in the class \#P, which includes NP.
 
 In terms of all these abbreviations, a summary of our MG circuit classical simulation results is given in Fig. 1.  These results resemble those for Clifford gates \cite{JoVa14}. In fact strikingly, if in Fig. 1 we change IN(PROD) to IN(BITS), and IN(MAGIC) to IN(PROD) (and in the body of Fig. 1 change ``Using SWAP gadget" to ``Using $T$ gadget") we obtain precisely the pattern of classical simulation complexities for Clifford circuits given in \cite{JoVa14}.

In order to derive these results, 
we utilise the two distinct techniques of simulating MG circuits that were introduced in~\cite{TeDi02} and~\cite{JoMi08}. One can be applied to simulate MG circuits with product state inputs and adaptive measurements (see \cite{Br16}) while the other works well for MG circuits with entangled inputs, but without adaptive measurements (see Observation \ref{obs:entangledinput}). It turns out that both of these can be extended to handle circuits containing ``a few'' disallowed elements in each of the techniques. In particular, we show how to simulate circuits with (i) many entangled input states and a few adaptive measurements (up to ${\cal O} (\log n)$), and (ii) circuits with many adaptive measurements and a few (a constant number of) entangled input lines (Theorems \ref{theo:fewadaptive} and \ref{theo:fewentangled}, respectively). Hence, adding a limited amount of the resources remains classically simulable; however, if too much of these resources are added, then the simulation method becomes classically infeasible. 

To prove the simulation results in Fig. 1 we utilise the techniques developed for Clifford circuits in \cite{JoVa14} and first show that the situation OUT(1), IN(PROD), ADAPT, STRONG is $\# P$--hard (see Theorem \ref{TH:sharpP}). This implies that the same setting except now with the more general input (IN(MAGIC)) is still $\# P$--hard. Similarly, allowing the more general output OUT(MANY), remains $\# P$--hard too (see rightmost column of Fig. \ref{fig:mg_magic_complexity}). Next we show that the simulation complexity of the NONADAPT case with OUT(MANY) and IN(MAGIC) is also $\# P$--hard (see Theorem \ref{thm:MIMOsharpP}). Finally to complete the table in Fig. 1, we prove that for the previous case, with STRONG replaced by WEAK, the existence of an efficient classical simulation would imply collapse of the Polynomial Hierarchy (PH) (see Theorem \ref{thm:phcollaps}), providing evidence for the non-existence of efficient weak simulation in this scenario.

Hence, apart form the previously known results (to which references are given in the body of Fig. \ref{fig:mg_magic_complexity}), it remained to show the orange shaded fields (which imply also the light-orange shaded fields as indicated by arrows) to completely characterise the simulation complexity of MG circuits for all of our considered scenarios.

\begin{figure}[!h]
\centering
\includegraphics[width=\columnwidth]{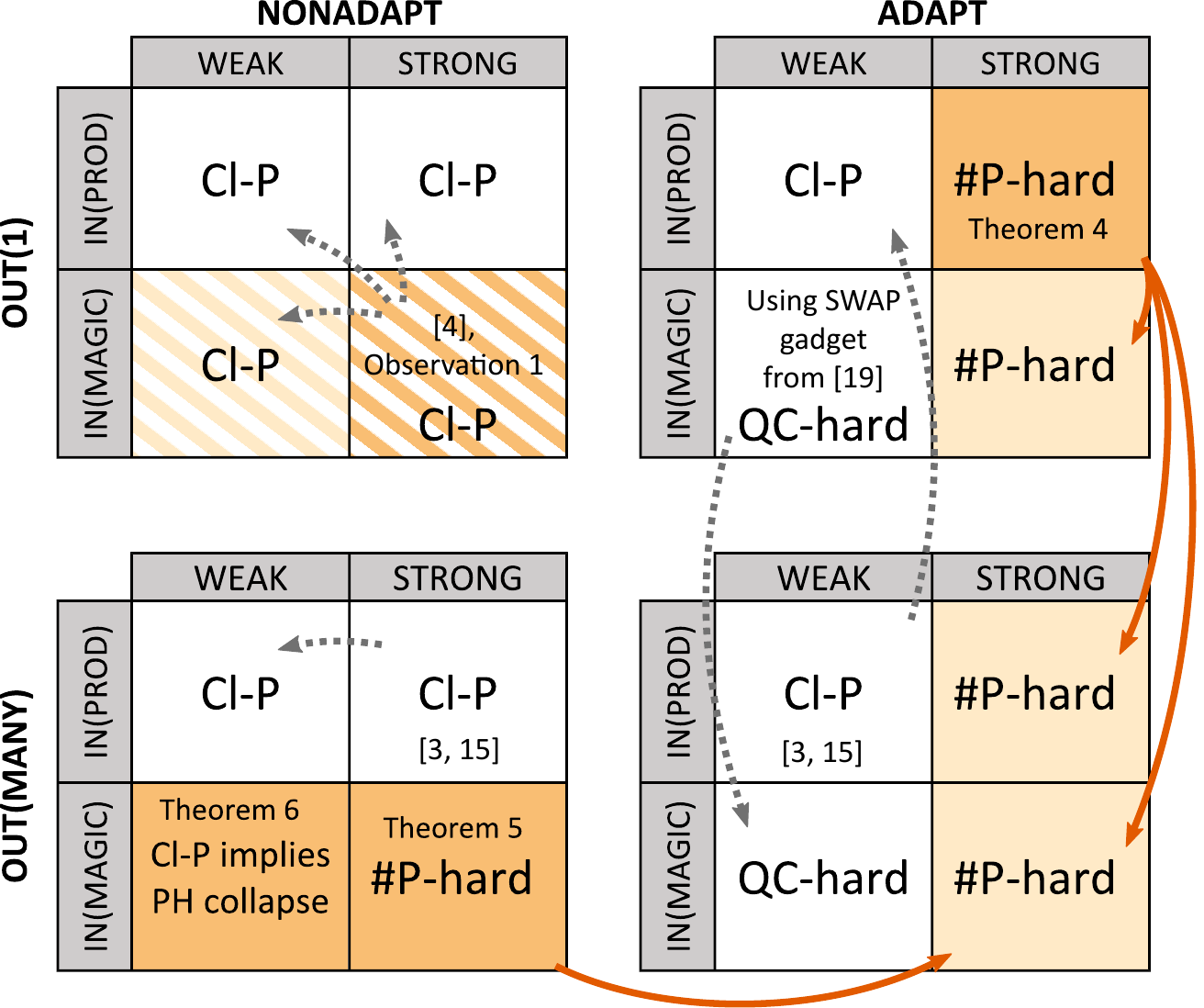}
\caption{Classical simulation complexity for different MG circuit scenarios. All abbreviations appearing are as defined in the main text. The table summarises previously known results together with our further results to provide the full picture. Statements that are proven explicitly here are shaded in orange. Fields shaded in light-orange may be inferred from these results. Arrows indicate some inferences that may be readily drawn. We use striped orange to indicate that the proof here (Observation \ref{obs:entangledinput}) is an easy corollary of results in \cite{JoMi08}.
The style of the table is borrowed from Figure 1 in \cite{JoVa14}, where Clifford circuits are studied. It is evident that this picture for MG circuits is formally identical to the one for Clifford circuits in \cite{JoVa14} after some reclassification of resources, as detailed in the main text.}
\label{fig:mg_magic_complexity}
\end{figure}

Finally we also consider classical simulation properties of two further settings for MG circuits, closely related to those considered above: firstly, when one has access to adaptive measurements in a basis different from the computational basis, and secondly, when one is able to use arbitrary entangled states as inputs.

In the first setting, it turns out that MG circuits yield universal quantum computation if in addition to measurements in the computational basis one allows adaptive measurements in the $\{\ket{+},\ket{-}\}$ basis. Moreover, this also holds if adaptive measurements are performed in any basis which is bounded away from the computational basis by an inverse polynomial in $n$ (Theorem \ref{theo:adaptarbitrary}).

Regarding the second setting, let us recall \cite{HeJo19} that any pure fermionic non--Gaussian state
(cf Section \ref{sec:preliminaries}A below for definitions)
 is a magic state for MG circuits, and also that any fermionic state of 3 or fewer qubits is Gaussian. This implies that usage of fermionic (generally entangled) states of 4 or more qubits in the input can elevate the computational power of adaptive MG circuits to quantum universal. Furthermore (cf \cite{HeJo19}) the state $\ket{\Phi^+}_{13}\ket{\Phi^+}_{24}$ of four qubits consecutively labelled as 1234 (and $\ket{\Phi^+}$ being the standard Bell state) is a magic state.
Here we show that MG circuits with arbitrary entangled 2-qubit inputs on {\em neighbouring} lines, adaptive measurements in the computational basis, and final measurements in any product basis, remain efficiently weakly simulable (Theorem \ref{theo:twoqubitinput}), generalizing the simulability result of \cite{Br16} (plus answering the open question (ii) therein). 
Furthermore as also already mentioned in \cite{HeJo19}, two copies of a non--fermionic 3-qubit entangled input states (on consecutive lines), such as the $GHZ$--state, can be utilised to generate 4-qubit magic states, which can then be used for universal quantum computation. Note that this does not contradict any previous findings because a single copy of the state $\ket{GHZ_3}$ cannot by itself be utilised directly to implement a gate giving universality with MGs; several copies have to be transformed first into each single magic state. Taken all together, these results completely characterise the classical simulability of MG circuits with respect to numbers of entangled lines in the input states.

\section{Preliminaries}
\label{sec:preliminaries}

\subsection{Notation}

We will write the standard qubit Pauli operators as $X,Y,Z$ and the identity operator as $\one$. 

A matchgate (MG) is a two--qubit unitary operator, $G(A,B)=A\oplus B$, where ${\rm det}\, A={\rm det}\, B$ and $A$ resp. $B$ acts non--trivially on the even resp. odd parity subspace of two qubits. Furthermore in MG circuits the 2-qubit matchgates are required always to act on {\em nearest neighbour} (n.n.) qubit lines only, and hereafter the term MG in any context involving many qubits will always implicitly refer to a n.n. action. Whereas a fermionic SWAP gate, fSWAP$=G(Z,X)$, is a MG, the SWAP gate, $SWAP=G(\one,X)$ is not a MG. It is known that supplementing MGs with the n.n. SWAP gate results in a quantum universal gate set \cite{JoMi08}. We mention that in contrast, in the setting of Clifford operations, the SWAP gate is already included as a Clifford gate and any 2-qubit Clifford operation may be applied to any pair of qubit lines.

MGs are well known to be closely related to the evolution of non--interacting fermions \cite{TeDi02}. We will say that a state is a {~\it fermionic state} if it has fixed parity, i.e. it is an eigenstate of the overall parity operator $Z \otimes \ldots \otimes Z$. This is the case if and only if when written in the computational basis, either all basis kets appearing contain an even number of 1s or all contain an odd number of 1s. Furthermore, we call a fermionic state a {\em Gaussian} fermionic state, if it can be generated by applying matchgates to a computational basis state.

The set of operations generated by hamiltonians which are arbitrary real linear combinations of the six hamiltonians $X_i X_{i+1}, Y_i Y_{i+1}, X_i Y_{i+1}, Y_i X_{i+1}, Z_i$ and $ Z_{i+1}$, coincides with the set of MGs acting on qubits $i$ and $i+1$. 
Let us also mention here that MGs admit a simple form when they are decomposed into the local and non-local part of two-qubit gates \cite{KrCi01}. In \cite{BrGa11} it has been shown that a two-qubit gate is a matchgate iff it can be written as
\begin{align}
\label{eq:mgdecomp}
\left( e^{i \phi_3 Z} \otimes e^{i \phi_4 Z} \right)   e^{i (\alpha X \otimes X + \beta Y \otimes Y)}    \left( e^{i \phi_1 Z} \otimes e^{i \phi_2 Z} \right)
\end{align}
 for some $\alpha, \beta, \phi_1, \phi_2, \phi_3, \phi_4 \in \mathbb{R}$.

\subsection{Classical simulation of quantum computations}

As mentioned above, we will distinguish between {\it strong} and {\it weak } classical simulation of a quantum computation. In order to recall the definitions we denote by $\{C_n\}$ a uniform family of quantum circuits (so they are poly-sized too)
acting on an $n$--qubit input state $\ket{\Psi_n}$. Suppose that the final measurement consists of measuring $k$ qubits (which might be $\mathcal{O}(n)$) in the computational basis. The probability of obtaining $x\in \{0,1\}^k$ is given by
\bea p(x\mid \Psi_n)= \bra{\Psi_n} C_n^\dagger\, \Pi_x \ C_n \ket{\Psi_n}           \eea 
where $\Pi_x$ is the projector onto the subspace of $n$ qubits spanned by all $n$-qubit computational basis states consistent with the given $k$-bit string $x$.

One says that the computation is classically efficiently simulable in the strong sense if each probability in this output distribution as well as each of its marginal probabilities, can be computed to precision of $m$ digits in classical poly$(n,m)$ time. We say that the computation is classically efficiently simulable in the weak sense if it is possible to sample from the output probability distribution in classical (randomised) poly$(n)$ time \cite{JoNi03,Va11}. 

Note that in the case of weak simulation, we can also obtain an estimate of the output probabilities as frequencies of outcomes in repeated samplings. According to the Chernoff-Hoeffding bound of probability theory (cf \cite{Va11}), $K$ samplings suffice to estimate a probability to a precision of $\mathcal{O}(\log K)$ digits with probability exponentially close (in $K$) to 1. Hence in classical poly$(n)$ time we obtain (with exponentially good probability) an estimate to precision of only $\mathcal{O}(\log n)$ digits, which is exponentially weaker than that demanded in a strong simulation (with $m$ being poly$(n)$ there).

\subsubsection{Classical simulation of MG circuits}
\label{sec:simreview}

We review here the two distinct techniques introduced in~\cite{TeDi02,Br16} and~\cite{JoMi08} for classically simulating MG circuits, that will be used in our results later. A reader familiar with the techniques for classically simulating MG circuits may skip the remainder of this subsection and continue at Section \ref{sec:magic4}.

First, we briefly recall some concepts important in the mentioned simulation techniques, and then focus on the scenario of OUT(1), and finally generalise the review to the case of OUT(MANY).

Consider a MG circuit on $n$ qubits. The $2n$ so-called Majorana operators (also called Jordan Wigner operators) are defined as the $n$-qubit operators (omitting tensor product symbols)
\begin{align*}
c_{2k-1} &= Z \ldots Z X \identity \ldots \identity\\
c_{2 k} &= Z \ldots Z Y \identity \ldots \identity,
\end{align*}
where the $X (Y)$ acts on the $k$-th line for $k = 1, \ldots , n$. These operators obey the anticommutation relations $\{c_\mu, c_\nu\} = 2 \delta_{\mu \nu} \identity$ for $\mu, \nu = 1,\ldots , 2n$. It will be convenient to also define a second set of operators fulfilling the anticommutation relations of fermionic creation and annihilation operators,
\begin{align*}
a_i = \frac{c_{2i-1} + i c_{2i}}{2} \qquad
a^\dagger_i = \frac{c_{2i-1} - i c_{2i}}{2}\\
\{a_i, a_j\} = \{a^\dagger_i, a^\dagger_j\} = 0,\  \{a_i, a^\dagger_j\} = \delta_{i j} \identity,
\end{align*}
where $i,j \in \{1, \ldots, n\}$.

The unitary action $U$ of any MG circuit can be written as $U=\operatorname{exp}(i \sum_{\mu \nu} h_{\mu\nu} c_\mu c_\nu)$, where $h$ is a real antisymmetric $2n \times 2n$  matrix, i.e. $U$ is generated by a Hamiltonian that is quadratic in the Majorana operators. Often, such a $U$ is called Gaussian. 
Crucial for the classical simulation of MG circuits is the fact that the linear span of the Majorana operators is preserved under conjugation by any MG circuit action $U$~\cite{TeDi02,JoMi08}:
\begin{align}
\label{equ:lin}
U c_\mu U^\dagger &= \sum_{\nu = 1}^{2n} R_{\mu, \nu} c_\nu, \nonumber\\
U^\dagger a_i U &= \sum_{\nu = 1}^{2n} T_{i, \nu} c_\nu, \\
U^\dagger a^\dagger_i U &= \sum_{\nu = 1}^{2n} T^*_{i ,\nu} c_\nu \nonumber.
\end{align}
(We remark in passing that interestingly~\cite{Jozsa}, this is in a sense formally similar to (but mathematically distinct from) the key property of Clifford operations underlying classical simulation results for them viz. that  the Pauli group is preserved under conjugation by Clifford operations.)
Here, $R$ in Eq.~(\ref{equ:lin})  can be easily determined from $U$, as $R=\operatorname{exp}(-4h)$ and $T$ is a complex $n \times 2n$ matrix defined by $T_{i, \nu} = \frac{1}{2}\left(R^T_{2i-1,\nu} + i R^T_{2i,\nu}\right)$. Using these relations it is then easy to see how MG circuits in the setting of IN(BITS) or IN(PROD) and OUT(1) can be classically efficiently strongly simulated (see Section \ref{sec:clsimsinglequbit}).

For the formalism that will underlie a second simulation method, let us now recall some concepts concerning the ordering of creation and annihilation operators leading to Wick's theorem \cite{TeDi02}.
Consider a product of creation and annihilation operators $A_1 A_2 \ldots A_k$, where each $A_i$ is either $a_j$ or  $a_j^\dagger$ for some $j \in \{1, \ldots, n\}$. Then, its normal-ordered form, which we denote by ${:}A_1 A_2 \ldots A_k{:}$, is defined as a rearrangement of the operators in the product $A_1 A_2 \ldots A_k$ such that all creation (annihilation) operators are on the left (right), but the ordering among the creation operators as well as the order among the annihilation operators is kept unchanged. Moreover, the contraction $\contraction{}{A_i}{}{A_j}A_i A_j$ of the pair of operators $A_i$ and $A_j$ is defined as $\contraction{}{A_i}{}{A_j}A_i A_j = A_i A_j  -  {:}A_i A_j{:}$. Clearly, $\contraction{}{a_i}{}{a_j}a_i a_j = \contraction{}{a_i}{}{a_j}a_i^\dagger a_j^\dagger = \contraction{}{a_i}{}{a_j}a_i^\dagger a_j = 0$ and $\contraction{}{a_i}{}{a_j}a_i a_j^\dagger = \delta_{i j} \identity$. A fundamental identity used in Wick's theorem is that an arbitrary product of creation and annihilation operators $A_1 A_2 \ldots A_k$ may be expressed as a sum over normal ordered forms of all possible contractions
\begin{align}
\label{eq:contractions}
A_1 A_2& \ldots A_k = \nonumber \\
  &{:}A_1 A_2 \ldots A_k{:} \nonumber\\
 +  &{:}\contraction{}{A_1}{}{A_2} A_1 A_2 \ldots A_k{:} +  {:}\contraction{}{A_1}{A_2}{A_3} A_1 A_2 A_3 \ldots A_k{:} + \ldots \nonumber\\
 +  &{:}\contraction{}{A_1}{}{A_2} A_1 A_2 \contraction{}{A_3}{}{A_4} A_3 A_4 \ldots A_k{:} +  {:}\contraction{}{A_1}{A_2}{A_3} \contraction[2ex]{A_1}{A_2}{A_3}{A_4} A_1 A_2 A_3 A_4\ldots A_k{:} + \ldots .\nonumber\\
 +  & \  \ldots
\end{align}
Importantly, the expectation values of normal ordered sequences of creation and annihilation operators (containing at least one such operator) under the state $\ket{\vec{0}}=\ket{0 \ldots 0}$ vanish. Thus, when considering the expectation value of an arbitrary sequence $A_1 A_2 \ldots A_k$ under the state $\ket{\vec{0}}$, only fully contracted terms (and there are potentially exponentially many of them) survive when utilising Eq.~(\theequation). This can be used to show that it is possible to efficiently calculate expectation values of such operator products under Gaussian fermionic states, a result known as Wick's theorem (see e.g. \cite{Br05Lagrange}). As explained in \cite{TeDi02}, when considering products of Majorana operators, i.e., $A_1 A_2 \ldots A_k$ where $A_i$  is some $c_\mu$, instead of the fermionic creation and annihilation operators, then Eq.~(\theequation) does not hold as an operator identity. Nevertheless, it does still hold when considering expectation values under the state $\ket{\vec{0}}$, with 
$\contraction{}{c_{\mu}}{}{c_{\nu}}c_{\mu} c_{\nu} = H_{\mu, \nu} \identity$, where $H$ is the block diagonal matrix  \cite{TeDi02}
\begin{equation} \label{heq}
H = \bigoplus_{j=1}^{n} \begin{pmatrix} 1 & i \\ -i & 1 \end{pmatrix}. 
\end{equation}

\subsubsection{Classical simulation of MG circuits: single qubit measurement}
\label{sec:clsimsinglequbit}

Both methods that we review here utilise the fact that $\proj{1}_1=1/2(\one+ Z_1)=1/2(\one+i c_1c_2)=a_1^\dagger a_1$. Hence, computing the output probability of a computational basis measurement on the first line of a MG circuit $U$ on the input state $\ket{\psi}$, leads to
\begin{align} \label{eq:p1}
p_1 &=\bra{\psi} U^\dagger a_1^\dagger U U^\dagger a_1 U \ket{\psi} \nonumber\\
&=\sum_{d,e} T_{1,d} T^*_{1,e}  \bra{\psi}c_e c_d \ket{\psi},
\end{align}
where we have introduced $\identity = U U^\dagger$ and used Eq.~(\ref{equ:lin}). 
This can be obviously generalised to measuring any single qubit $k$. 

The first method to simulate MG circuits which we recall here, has been introduced in \cite{JoMi08}. In the following, we will call this technique also the `Heisenberg technique'.
It makes use of the fact that efficient evaluation of the probabilities of the individual outcomes (strong simulation) is possible when the input state is a product state. The sum in Eq.~(\ref{eq:p1}) has only $\mathcal{O}(n^2)$ terms, and moreover, as the input state is a product state, the summands can be calculated efficiently as they factorise into local operators (as the $c_k$ are local operators). Note also that here and in the following, the matrices $R$ and $T$ associated to the unitary $U$, which describes the full MG circuit, may be efficiently determined by forming sequential products of the matrices $R_i$ and $T_i$ corresponding to the poly-many individual MGs comprising the circuit.

The second method to compute the output of the final single qubit measurement is based on Wick's theorem \cite{TeDi02}. There, one uses the fact that the expression in Eq.~(\ref{eq:p1}) may be rewritten as the Pfaffian of a matrix. In this method it is important that the input $\ket{\psi}$ is a computational basis state $\ket{w}$, where $w \in \{0,1\}^n$ rather than a more general product state (although this restriction was later circumvented in \cite{Br16}). The input state is then expressed as a product of Majorana operators $c_{2p_1-1}, \ldots, c_{2p_l-1}$ acting on the state $\ket{\vec{0}}$, where $l$ is the Hamming weight of the bitstring $w$. Eq.~(\ref{eq:p1}) then reads
\begin{align}
p_1 = &\sum_{d,e}  T_{1,d} T^*_{1,e} \nonumber\\
  &\ \times \bra{\vec{0}} c_{2p_l-1} \ldots c_{2p_1-1} c_e c_d c_{2p_1-1} \ldots c_{2p_l-1}\ket{\vec{0}}.
\end{align}
In Eq.~(\theequation), the product of Majorana operators can be replaced by the sum over all fully contracted terms as explained above (see Eq. (\ref{eq:contractions})). It has been shown that the expression in Eq.~(\theequation) can then be expressed as the Pfaffian of an antisymmetric $(2l+2) \times (2l+2)$-matrix $O$ \cite{TeDi02}\footnote{Note that as shown in \cite{TeDi02}, the formula obtained from the contractions depends only on the entries $O_{ij}$ for $i<j$ and defining then an antisymmetric matrix $O$ allows to rewrite this expresssion as a Pfaffian.}. 
The entries $O_{ij}$ may be constructed with the help Table II in \cite{TeDi02}, which we reprint here as Table \ref{tab:lookup1qubit}. A Pfaffian arises here since the sum over the contractions corresponds to the sum of the signed permutations of the matrix elements $O_{ij}$. Due to the sum over $d$ and $e$, these matrix elements can be computed as the matrix products given in Table \ref{tab:lookup1qubit}. More precisely, the entry $O_{ij}$ may be constructed as follows. First, one looks up the subscript labels of the Majorana operators in the $i$th and $j$th factor (counting from the left) within the product of $2l+2$ Majorana operators in Eq.~(\theequation), obtaining a $(2p-1)$- or a $d$- or an $e$- label for each of $i$ and $j$. 
Then, $O_{ij}$ may be read off from Table \ref{tab:lookup1qubit} according to the determined label types. Note that there exists only one Majorana-operator with $d$- and $e$-labels in Eq.~(\theequation) and the index 1 is omitted. Hence, for now, some of the entries in Table \ref{tab:lookup1qubit} are not required. Moreover, $i_\alpha$ and $i_\beta$ will be used to label the measured lines later on; for now $i_\alpha = i_\beta = 1$.  
Importantly, the matrix $O$ may be efficiently constructed and, moreover, as the Pfaffian of any antisymmetric matrix $A$ fulfills $\operatorname{Pf}(A)^2 = \det{(A)}$, $\operatorname{Pf}(O)$ can be efficiently computed.

\begin{table}[hbt]
\begin{tabular}{clllllll}
\hline
\multicolumn{1}{l}{} & \multicolumn{4}{c}{j}               \\ \hline
\multicolumn{1}{l}{} &                                          &  $c_{d_\beta}$                                                                         &  $c_{e_\beta}$                                                                                     &  $c_{2p_\beta-1}$  \\ \hline
\multirow{3}{*}[-5pt]{i}    &   $c_{d_\alpha}$              &  $\left(TH{T}^T\right)_{i_\alpha, i_\beta}$         &  $\left(TH{T}^\dagger\right)_{i_\alpha, i_\beta}$         &  $\left(TH\right)_{i_\alpha, 2p_\beta-1}$ \\
                                 &  $c_{e_\alpha}$              &  $\left({T}^*H{T}^T\right)_{i_\alpha, i_\beta}$ &  $\left({T}^*H{T}^\dagger\right)_{i_\alpha, i_\beta}$  &  $\left({T}^*H\right)_{i_\alpha, 2p_\beta-1}$\\
                                 &  $c_{2p_\alpha-1}$            &  $\left( H T^T\right)_{2 p_\alpha-1, i_\beta}$                            & $\left( H T^\dagger \right)_{2 p_\alpha-1, i_\beta}$                            &  $\delta_{\alpha, \beta}$ \\ \cline{1-5} 
\end{tabular}
\caption{Lookup table to construct a matrix $O_{ij}$ for $i<j$ as in \cite{TeDi02}. Here $H$ is as in Eq.~(\ref{heq}). Moreover, $\alpha$ denotes the $\alpha$th qubit which is measured (in case the operator of interest is $c_{d_\alpha}$ or $c_{e_\alpha}$), or initially prepared in the state $\ket{1}$ (in case the operator of interest is $c_{2p_\alpha-1}$), and similarly for $\beta$. As explained in the main text, considering the Pfaffian of the matrix $O$ will then allow the simulation of the MG circuit associated to $T$.}
\label{tab:lookup1qubit}
\end{table}

\subsubsection{Simulation of MG circuits involving more measurements}
\label{sec:adaptive}

We now generalise Eq.~(\ref{eq:p1}) to multi--qubit measurements. We first consider a multi-line final measurement and then include the possibility of adaptive measurements during the computation. Let us denote by $i_1 \ldots i_k$ the $k$ lines which are measured in the computational basis yielding a $k$-bit string $x=x_1\ldots x_k$, and let $U$ be the overall unitary action of the MG circuit. Using Eq.~(\ref{equ:lin}) and the fact that for each line (omitting line subscript labels) projectors are given by
\begin{equation}\label{projects}
\Pi(0)=\proj{0}=aa^\dagger \hspace{5mm}\Pi(1)=\proj{1}=a^\dagger a, 
\end{equation}
we see that the probability to obtain $x$ is given by
\begin{align}
p(x) =\sum_{d_1,e_1, \ldots, d_k, e_k}  T_{i_1,d_1}  T^*_{i_1,e_1} \ldots   T_{i_k,d_k}  T^*_{i_k,e_k} \nonumber\\ \times \bra{\psi} A_{d_1 e_1}(x_1)\ldots A_{d_k e_k}(x_k) \ket{\psi},
\end{align}
where 
\begin{equation}\label{As}
A_{de}(0)=c_d c_e \hspace{5mm} A_{de}(1)=c_e c_d.
\end{equation}

Note that the number of lines measured might be $\mathcal{O}(n)$ so that the sum in Eq.~(\theequation) might involve exponentially many terms. 
Thus generally it is not possible any more to evaluate all terms in the sum individually and to employ the Heisenberg technique for an efficient classical simulation. However, here the strength of the second simulation method \cite{TeDi02} comes into play. Given that $\ket{\psi}$ is a computational basis state, Eq.~(\theequation) may still be rewritten as the Pfaffian of an efficiently constructible matrix $O$, which is now a $[2(l+k)] \times [2(l+k)]$-matrix. This matrix may be explicitly constructed with the help of Table \ref{tab:lookup1qubit} as explained above (and cf \cite{TeDi02} for further details).

\begin{figure}[b]
\includegraphics[width=0.8\columnwidth]{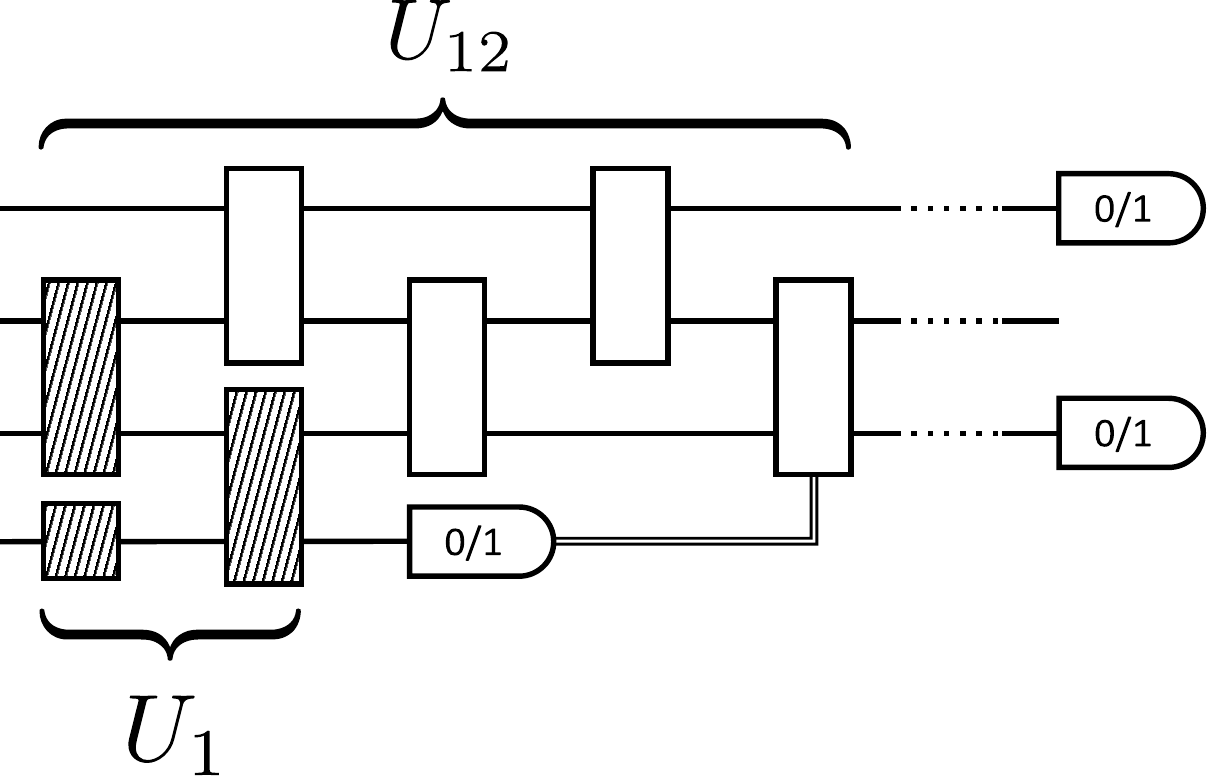}
\caption{Sketch of a MG circuit involving adaptive measurements. By $U_1$, we denote the part of the circuit needed up to the first adaptive measurement (shaded MGs) and by $U_{12}$ we denote the unitary corresponding to the gates needed up to the next meassurement (all the gates shown). $U_{12}$ depends on the outcome of the first adaptive measurement.
As explained in the main text, the strategy to simulate such a circuit is the following. First, simulate only the part up to the first adaptive measurement (shaded MGs); sample an outcome; replace the first adaptive measurement by an projector onto the obtained outcome; simulate the circuit up to the next measurement, and so on.}
\label{fig:adaptive}
\end{figure}
The next step is to treat the possibility of adaptive measurements in the computational basis~\cite{TeDi02,Br16}. We consider each adaptive measurement to be a single line computational basis measurement. Let us denote the sequence of lines that are adaptively measured  and lines measured at the end of the computation respectively by $J=(j_1, j_2, \ldots)$ and $I = (i_1, i_2, \ldots)$. Let us now determine the joint probability of obtaining a given bit string of intermediate measurement outcomes $y=y_1 \ldots y_{|J|}$ and the final measurement outcome string $x=x_1 \ldots x_{|I|}$. The expression is obtained by replacing all measurements by projectors which are then expressed in terms of the fermionic creation and annihilation operators for each specified outcome using Eq.~(\ref{projects}) as before. Now, the order of the measurements and gates must be taken into account. 
Introduce the notations
\begin{gather*}
 \Pi_{i_1\ldots i_{|I|}}(x)=\Pi_{i_1}(x_1)\ldots \Pi_{i_{|I|}}(x_{|I|}) \\
 A_{a_1b_1, \ldots ,a_{|I|}b_{|I|}}(x) = A_{a_1b_1}(x_1)\ldots A_{a_{|I|}b_{|I|}}(x_{|I|}).
 \end{gather*}

In the case of a single adaptive measurement (i.e. $|J|=1$ and $y=y_1$) we obtain
\begin{align}
p(x, y_1) &=  \\ 
\bra{\psi} U_1^\dagger \Pi_{j_1}(y_1) U_2^\dagger 
\Pi_{i_1\ldots i_{|I|}}(x) \nonumber
&  U_2 \Pi_{j_1}(y_1) U_1 \ket{\psi},
\end{align}
where $U_1$ is the unitary corresponding to all the gates that need to be applied up to the first adaptive measurement, $U_2$ is the unitary of all the gates needed between the first and the next measurement (in the present case of $|J|=1$ so $U_2$ comprises of all the remaining gates) etc. 
Note also that in an adaptive circuit, $U_2$ generally depends on the specified value $y_1$. However, in order to simplify notation, we will not explicitly label such dependences, here and in the following. Let us now additionally introduce the notation $U_1$, $U_{12}$, etc, where $U_1$ is as before, $U_{12}$ corresponds to all the gates from the beginning of the circuit to the second adaptive measurement, and so on (see Figure \ref{fig:adaptive}). Again, $U_{12}$ (and all following unitaries) will generally depend on the outcome of the previous adaptive measurements. We denote the $T$-matrices (from Eq.~(\ref{equ:lin})) corresponding to the respective circuit parts by $T^1, T^{1 2}, $ etc. The necessity of this notation results from the necessity of introducing $U_1 {U_1}^\dagger$, $U_{12} {U_{12}}^\dagger$, etc. at the appropriate places in Eq.~(\theequation) \cite{TeDi02}, in order to make use of Eq.~(\ref{equ:lin}). We display the resulting expression for the joint probability for the present simple case of a single adaptive measurement ($|J|=1$) in Eq.~(\ref{eq:jointprob}) below, as well as the general case of an arbitrary number of adaptive measurements \cite{TeDi02} in Eq.~(\ref{eq:jointprobmany}).   
 
\begin{widetext}
\begin{align}
\label{eq:jointprob}
p(x,y_1) = \sum_{\substack{a_1,b_1,f_1,g_1,\\d_1,e_1, \ldots, d_{|I|},e_{|I|}}}& T^1_{j_1,a_1} T^{1*}_{j_1,b_1}  T^{12}_{i_1,d_1} T^{12*}_{i_1,e_1} \ldots T^{12}_{i_{|I|},d_{|I|}} T^{12*}_{i_{|I|},e_{|I|}}   T^1_{j_1,f_1} T^{1*}_{j_1,g_1} \nonumber\\
&\qquad\times \bra{\psi} A_{a_1b_1}(y_1)  A_{d_1e_1,\ldots  , d_{|I|}e_{|I|}} (x) A_{f_1g_1}(y_1) \ket{\psi},\\
\label{eq:jointprobmany}
p(x,y) = \sum_{\substack{a_1,b_1,f_1,g_1,\ldots,\\ a_{|J|},b_{|J|},f_{|J|},g_{|J|},\\d_1,e_1, \ldots, d_{|I|},e_{|I|}}}    
 &T^1_{j_1,a_1} T^{1*}_{j_1,b_1} \ldots T^{1 2 \ldots|J|}_{j_{|J|},a_{|J|}} T^{1 2 \ldots |J|*}_{j_{|J|},b_{|J|}} T^{1 2 \ldots |J| + 1}_{i_1,d_1}T^{1 2 \ldots |J|+1*}_{i_1,e_1} \ldots T^{1 2 \ldots |J|+1}_{i_{|I|},d_{|I|}} T^{12 \ldots |J|+1*}_{i_{|J|},e_{|J|}}   \nonumber\\ 
 & \qquad \times T^{1 2 \ldots |J|}_{j_{|J|},f_{|J|}} T^{1 2 \ldots |J|*}_{j_{|J|},g_{|J|}} \ldots T^1_{j_1,f_1} T^{1*}_{j_1,g_1}\nonumber \\ &\qquad\times \bra{\psi} 
 A_{a_1b_1,\ldots ,a_{|J|}b_{|J|}}(y) A_{d_1e_1,\ldots  , d_{|I|}e_{|I|}} (x) 
  A_{a_{|J|}b_{|J|},\ldots ,a_1b_1}(y)
  \ket{\psi}.
\end{align}
\end{widetext}

 What we are actually interested in is the probability distribution corresponding to only the final measurement $p(x)$. 
Considering again the case $|J|=1$, the probability $p(x,y_1)$ is the joint probability of observing the intermediate measurement outcome $y_1$ and the final measurement outcome $x$.
 But introducing $p(x | y_1 )$, the conditional probability of observing $x$ given that the outcome of the intermediate measurement was $y_1$, we have that
\begin{align}
p(x) = \sum_{y_1} p(x | y_1 ) p(y_1).
\end{align}
Similarly, in case of $|J|$ intermediate measurements we have
\begin{align}
&p(x)  = \sum_{y_1, \ldots, y_{|J|}} p(x, y_{|J|}, \ldots, y_1) \nonumber\\
& \qquad = \sum_{y_1, \ldots, y_{|J|}}  p(x | y_{|J|}, \ldots, y_1)   \nonumber \\
& \qquad \qquad  \times   p(y_{|J|} | y_{|J|-1}, \ldots, y_1) \ldots    p(y_2 | y_1) p(y_1),
\end{align}
where we have iteratively expressed $p(x)$ in terms of the conditional and the marginal probability distributions.

With this, we are now ready to recall the method to weakly simulate adaptive MG circuits with computational basis input \cite{TeDi02}. The method is based on two key ingredients. First, under certain conditions which we discuss later on, expressions of the form in Eq.~(\ref{eq:jointprobmany}) may be efficiently computed. Second, a sample from the probability distribution $p(x)$ may be obtained through iteratively sampling marginal distributions as in Eq.~(\theequation). The procedure to do so is the following \cite{TeDi02}: (1) the circuit up to the first intermediate measurement is considered and the probabilities for outcomes $y_1$, $p(y_1)$ are computed; (2) one classically samples one outcome $y_1$ from $p(y_1)$ and fixes this outcome for the rest of the procedure; (3) then, $p(y_2 | y_1)$ is calculated using $p(y_2,y_1)/p(y_1)$, where the numerator can be calculated by evaluating Eq.~(\ref{eq:jointprobmany}) and the denominator has been determined previously; (4) an outcome of the second measurement $y_2$ is classically sampled from $p(y_2 | y_1)$; (5) the preceding steps are repeated until the final measurement outcome $x$ is sampled.

Let us now recall why Eq.~(\ref{eq:jointprobmany}) may be calculated efficiently in case of computational basis input \cite{TeDi02}. It is clear that the number of summands in Eq.~(\ref{eq:jointprobmany}) is exponential in the number of intermediate measurements $|J|$ as well as in the number of final measurements $|I|$.
Thus, in general it is impossible to evaluate all terms in the sum individually. However, as before, the sum may be reexpressed as the Pfaffian of an antisymmetric square matrix $O$, which here is of dimension $2(l+|I| + 2 |J|)$, where $l$ is the Hamming weight of the input string. Recall also that $|I|$ is the number of final measurements and that $|J|$ is the number of adaptive measurements. For details on constructing $O$ in this case, in particular, for the required counterpart of Table \ref{tab:lookup1qubit}, we refer the reader to \cite{TeDi02}. Note that the technique remains efficient for $\mathcal{O}(\operatorname{poly}(n))$ adaptive measurements.

\subsubsection{Comparing simulation techniques and extensions} 
\label{sec:hadamardgadget}

The classical simulation techniques introduced in \cite{JoMi08} can be straightforwardly generalised to input states that contain (arbitrarily many) groups of qubits that are initialised in an arbitrary entangled state, as long as the size of each group is at most $\mathcal{O}(\log n)$. We will elaborate on that in Section \ref{sec:clsim}. In this context it does not matter how the entangled states are distributed over the lines, in particular, the lines with entangled states need not be adjacent. Interestingly, the second technique described above (and as presented in (\cite{TeDi02}) cannot provide efficient classical simulation in this scenario of more general input states.

The method of efficiently simulating adaptive measurements introduced in \cite{TeDi02} has been generalised to apply to an arbitrary product state input in \cite{Br16}. It relies on the efficient generation of any product state from the input state $\ket{0}^{\otimes n} \ket{+}$ by a MG circuit, and then showing that adaptive MG circuits on the ``almost-computational-basis" state $\ket{0}^{\otimes n} \ket{+}$ remain efficiently classically simulable (by considering the two superposition components of $\ket{+}$ separately). We will use this technique later in our work and briefly review it here. The reduction is achieved by making use of the so-called Hadamard gadget, which allows one to transform $\ket{0}^{\otimes n} \ket{+}$ into any desired product state using a MG circuit. The idea is to apply $G(H,H)$ to the auxiliary line, initialised in the state $\ket{+}$, and its n.n. target line. This leaves the auxiliary line unchanged as $\ket{+}$ and acts as a Hadamard gate on the target line. Together with single qubit phase gates, which are matchgates, this allows for the construction of an arbitrary single qubit unitary on a target line. Using fermionic swaps, $fSWAP = G(Z,X)$, the resulting state is then swapped from the target line to the required position. Fermionic swaps act identically to the conventional swap gate when any of the two input states is $\ket{0}$. Hence, to generate an arbitrary input state starting with $\ket{0}^{\otimes n} \ket{+}$, one can proceed as follows \cite{Br16}. First, $\ket{\phi_1}$ is created on line $n$, then $n-1$ fermionic swaps are used to swap the created state through the first $n-1$ lines (which are all in $\ket{0}$). The overall state is now $\ket{\phi_1}\ket{0}^{\otimes n-1} \ket{+}$. Then, $\ket{\phi_2}$ is created on line $n$ and swapped into position 2 using $n-2$ fermionic swaps and so on.

\subsection{Magic states for MG circuits}
\label{sec:magic4}
The notion of \emph{magic state} was first introduced in  \cite{BrKi05} in the context of Clifford circuits. Extending Clifford gates with the $T$ gate gives a universal gate set. But instead of enlarging the gate set, one can also consider allowing more general input states (called magic states) and adaptive measurements to implement a new gate by using only previously available gates. Specifically, for the $T$ gate this is achieved by the so-called $T$-gadget, a small adaptive Clifford circuit  that consumes one copy of the magic state $\ket{A} = 1/\sqrt{2} \left(\ket{0} + e^{i \pi/4} \ket{1}\right)$ as well as one adaptive measurement in the computational basis to allow deterministic implementation of one $T$-gate \cite{BrKi05}.
However, neither copies of the magic state, nor adaptive measurements in the computational basis on their own give rise to universal quantum computation. On the contrary, these situations have been proven to be classically efficiently simulable~\cite{JoVa14}.

In earlier work, we studied the notion of magic states in the context of MG circuits in~\cite{HeJo19}. It turns out that its definition is more subtle there compared to the context of Clifford circuits, due to the n.n. condition for MG actions and non-availability of the SWAP operation. Also since MG circuits remain classically simulable even if arbitrary product input states and adaptive measurements in the computational basis are allowed \cite{Br16} we see that  single-qubit states cannot be magic states for MG circuits.
In \cite{HeJo19}, we introduced the following natural definition of magic states for MG circuits:  if $R$ is a resourceful $k$-qubit gate (i.e. giving universal computation with MGs), we say that an $m$-qubit state $\ket{M}$ is a {\em magic state} for $R$ if\\
(M1): there is a circuit $C$ of MGs and adaptive measurements such that for any $k$-qubit state $\ket{\alpha}$, $C$ maps $\ket{\alpha}\ket{M}$ to $(R\ket{\alpha})\ket{\tilde{M}}$ (where $\ket{\tilde{M}}$ is any state, that may depend on the intermediate measurement outcomes too. However, it may not depend on $\ket{\alpha}$.)\\
(M2): The state $\ket{M}$ can be swapped through arbitrary states using only MGs.\\
(Actually, a slightly more general version of (M1), tolerating a small error, is used \cite{HeJo19} but for transparency we will not reproduce it here).

We showed that (M2) is fulfilled iff the state is fermionic, and we derived the following characterisation of pure magic states. 
\begin{theorem}[\cite{HeJo19}]
\label{theo:allnongaussianmagic}
Any pure fermionic state which is non--Gaussian is a magic state for MG computations.
\end{theorem}

An example of a magic state is the 4--qubit state $\ket{M}=\ket{\Phi^+}_{13}\ket{\Phi^+}_{24}$ (where $\ket{\Phi^+}$ is the standard Bell state), which can also be reversibly transformed via MGs into the 4--qubit GHZ-state \cite{HeJo19,Br06}. Similar to the $T$--gadget for Clifford operations, these states can be used to implement the SWAP operation (known to be resourceful for MGs) on two neighbouring lines as detailed in \cite{HeJo19}; briefly, first, a copy of $\ket{M}$ is swapped in between the two target lines using fSWAPs. Then, two Bell measurements are performed on each of the two target lines and the outer lines of $\ket{M}$, respectively, followed by adaptive Pauli corrections. A Bell measurement may be implemented through the MG $G(H,H)$ followed by computational basis measurements. Finally, the measured lines are swapped to the bottom of the circuit using $G(\pm Z, X)$s. All of this amounts to an adaptive MG circuit that deterministically implements a n.n. SWAP gate.

\section{Classical simulability of MG circuits supplemented with entangled state inputs and/or adaptive measurements}
\label{sec:Mainresults}
\label{sec:clsim}

We have seen that MG circuits are classically simulable (even in the strong sense) in the setting of IN(BITS) and OUT(1), and above in Section \ref{sec:magic4}, that MG circuits supplied with copies of the magic input state $\ket{M}$ and adaptive measurements in the computational basis give rise to universal quantum computation. In this section, we investigate the case of OUT(1) where one has access to either of the resources ADAPT or IN(MAGIC). The case IN(PROD), ADAPT, OUT(1) has already been shown to be weakly simulable in \cite{Br16}. Hence for OUT(1), the only missing case is IN(MAGIC), NONADAPT, OUT(1). We show that this scenario is strongly simulable (Observation \ref{obs:entangledinput}). After that, we analyse the simulation complexity in case one has limited access to the resources. In Sec. \ref{sec:hardnessclassification} we will then consider further settings including those with OUT(MANY), to complete demonstrations of all results given in Figure 1.

\subsection{Many entangled states and/or many adaptive measurements}

The case in which magic states are available, but no adaptive measurements are allowed is classically simulable due to the Heisenberg technique \cite{JoMi08}, as the following observation, which is a straight forward generalisation of \cite{JoMi08}, shows.
\begin{observation}
\label{obs:entangledinput}
A matchgate circuit with product state input and a final single qubit measurement, and no adaptive measurements, which is additionally supplemented by magic state input in form of copies of $\ket{M}$ (and in fact if supplemented by arbitrary entangled states, each involving up to $\mathcal{O}(\log n)$ lines) is classically efficiently simulable in the strong sense.
\end{observation}
\begin{proof}
Consider the proof in~\cite{JoMi08}, which we briefly recalled in Section \ref{sec:simreview}, which shows the statement for product state inputs. The proof therein generalises to an entangled input state $\ket{\psi}$, as long as the input state can be written as a tensor product $\ket{\psi} = \ket{\phi_1} \otimes \ldots \otimes \ket{\phi_l}$, where each $\ket{\phi_i}$ involves only up to $\mathcal{O}(\log n)$ (not necessarily contiguous) qubits. In fact, in Eq.~(\ref{eq:p1}) the matrices involved in the computation would be of size $\mathcal{O}(\operatorname{poly} n)$.
\end{proof}

Table \ref{tab:simcost} summarises the cost of simulating MG circuits depending on whether magic state input and/or adaptive measurements in the computational basis are allowed (see also Figure \ref{fig:mg_magic_complexity}). The first two rows in Table \ref{tab:simcost}  are classically simulable for different reasons. The case where adaptive measurements in the computational basis are allowed was addressed in~\cite{Br16}.

\begin{table}[ht]
\begin{tabular}{lll}
\textbf{IN(MAGIC)} & \textbf{ADAPT} & \textbf{Simulation cost}
\\ \hline
Yes                  & No                             & Cl-P~\cite{JoMi08}, Obs. \ref{obs:entangledinput}  \\
No                   & Yes                            & Cl-P~\cite{TeDi02,Br16}     \\
Yes                  & Yes                            & QC-hard (SWAP-gadget)~\cite{HeJo19}
\end{tabular}
\caption{Simulation cost of MG circuits and available resources: adaptive
measurements in the computational basis and magic input states $\ket{M}$ are
available. `Yes' indicates that a polynomial amount of the resource is
available. `No' means that a resource is not available.}
\label{tab:simcost}
\end{table}

\subsection{Limited number of magic states or adaptive measurements}

Here we will show that one can extend the classical simulation techniques to allow `a few' of the resources at the places where we find `no' resources in Table \ref{tab:simcost} retaining efficient classical simulability. 

\begin{theorem}
\label{theo:fewadaptive}
MG circuits on $n$ qubits initialized in an input state $\ket{\psi}$ consisting of arbitrary entangled states $\ket{\psi_i}$ on at most $\mathcal{O}(\log n)$ qubits each, supplemented by $k$ (single-line) adaptive measurements as well as a single final measurement in the computational basis can be (weakly) classically simulated in $\mathcal{O}(\operatorname{poly}(n) (2n)^{4k})$ time.
\end{theorem}
\begin{proof}
Consider the expression in Eq.~\eqref{eq:jointprobmany}. It is clear that in case of an input state as described in the theorem, evaluating a single summand involves multiplication of $O(n)$-dimensional vectors and a matrix which is a tensor product of $\mathcal{O}(\log n)$ Pauli matrices and thus takes $\mathcal{O}(\operatorname{poly}(n))$ time. Now note that the number of summation variables occurring in Eq.~\eqref{eq:jointprob} is given by two times the number of finally measured lines (here, 1) plus $4k$, and each of the variables runs from $0$ to $2n-1$. Thus, there are $(2n)^{4k+2}$ summands. In order to obtain a sample from the final measurement, actually $k+1$ simulations are performed, where the $i$th simulation simulates the circuit up to the $i$th measurement (cf. Section \ref{sec:adaptive}). However, as $k$ is bounded by $\operatorname{poly}(n)$ (recalling that all our circuits are poly-sized), one thus obtains a total simulation cost of $\mathcal{O}(\text{poly}(n) (2n)^{4k})$, which proves the theorem.
\end{proof}

Note that the simulation in Theorem \ref{theo:fewadaptive} is efficient if we only have up to a constant number of intermediate measurements, $k$. Moreover, $\mathcal{O}(\log n)$ adaptive measurements in the computational basis lead to a quasi-polynomial runtime of the classical simulation algorithm.

Our next result establishes an upper bound on the classical simulation cost of the setting of poly-many adaptive measurements and entangled inputs. Note that the adaptive setting considered here encompasses OUT(1) and OUT(MANY).

\begin{theorem}
\label{theo:fewentangled}
MG circuits with product state input on $n$ lines, supplemented with an arbitrary entangled input state on $k$ lines, and poly$(n)$ many adaptive measurements in the computational basis, can be (weakly) classically simulated in $\mathcal{O}(\text{poly}(n) 2^{2k})$ time.
\end{theorem}

\begin{proof}

\begin{table*}[t]
\begin{tabular}{clllllll}
\hline
\multicolumn{1}{l}{} & \multicolumn{7}{c}{j}               \\ \hline
\multicolumn{1}{l}{} &                        & $c_{2q_\beta-1}$ &  $c_{(a/f)_1}$                                                                      &  $c_{(b/g)_1}$                                                                               &  $c_{d_\beta}$                                                                 &  $c_{e_\beta}$                                                                         &  $c_{2p_\beta-1}$  \\ \hline
\multirow[v]{6}{*}[-18pt]{i}& $c_{2q_\alpha-1}$& 0                         &  $\left(H{T^1}^T\right)_{2q_\alpha-1, j_1}$                  &  $\left(H{T^1}^\dagger\right)_{2q_\alpha-1, j_1}$         &  $\left(H{T^{12}}^T\right)_{2q_\alpha-1, i_\beta}$                &  $\left(H{T^{12}}^\dagger\right)_{2q_\alpha-1, i_\beta}$                &  $\delta_{q_\alpha,p_\beta}$ \\
                           &   $c_{(a/f)_1}$     & X                         &  $\left(T^1H{T^1}^T\right)_{j_1, j_1}$                      &  $\left(T^1H{T^1}^\dagger\right)_{j_1, j_1}$                      &  $\left(T^1H{T^{12}}^T\right)_{j_1, i_\beta}$             &  $\left(T^1H{T^{12}}^\dagger\right)_{j_1, i_\beta}$             &  $\left(T^1H\right)_{j_1, 2p_\beta-1}$ \\
                           &   $c_{(b/g)_1}$    & X                         &  $\left({T^1}^*H{T^1}^T\right)_{j_1, j_1}$               &  $\left({T^1}^*H{T^1}^\dagger\right)_{j_1, j_1}$              & $\left({T^1}^*H{T^{12}}^T\right)_{j_1, i_\beta}$      &  $\left({T^1}^*H{T^{12}}^\dagger\right)_{j_1, i_\beta}$       &  $\left({T^1}^*H\right)_{j_1, 2p_\beta-1}$\\
                           &   $c_{d_\alpha}$  & X                         &  $\left(T^{12}H{T^1}^T\right)_{i_\alpha, j_1}$          &  $\left(T^{12}H{T^1}^\dagger\right)_{i_\alpha, j_1}$         &  $\left(T^{12}H{T^{12}}^T\right)_{i_\alpha, i_\beta}$         &  $\left(T^{12}H{T^{12}}^\dagger\right)_{i_\alpha, i_\beta}$         &  $\left(T^{12}H\right)_{i_\alpha, 2p_\beta-1}$ \\
                           &  $c_{e_\alpha}$   & X                         &  $\left({T^{12}}^*H{T^1}^T\right)_{i_\alpha, j_1}$  &  $\left({T^{12}}^*H{T^1}^\dagger\right)_{i_\alpha, j_1}$ &  $\left({T^{12}}^*H{T^{12}}^T\right)_{i_\alpha, i_\beta}$  &  $\left({T^{12}}^*H{T^{12}}^\dagger\right)_{i_\alpha, i_\beta}$  &  $\left({T^{12}}^*H\right)_{i_\alpha, 2p_\beta-1}$\\
                           &  $c_{2p_\alpha-1}$ & X                         &  X                                                                                           &  X                                                                                                     &  X                                                                                    &  X                                                                                             &  0 \\ \cline{2-8}
\end{tabular}
\caption{Generalised lookup table for the construction of matrix elements $O_{i,j}$ for $i<j$, which are required in the proof of Theorem \ref{theo:fewentangled} (based on the technique for simulating MG circuits introduced in \cite{TeDi02}). `X' indicates that such terms do not appear. We use the same notation as in Section \ref{sec:adaptive}, e.g., $T^{12}$ is associated to the gates up to the second measurement (see also Figure \ref{fig:adaptive}). Moreover, $\alpha$ denotes the $\alpha$th qubit which is finally measured (for cases $c_{d_\alpha}$, $c_{e_\alpha}$), or contains `1' in the currently considered computational basis component $w$ (for case $c_{2p_\alpha}$) or $w'$ (for case $c_{2q_\alpha}$), and similarly for $\beta$. 
 Despite the fact that the first column (last row) could be omitted, we give them nevertheless to make it clear that terms of that form either do not appear or do not contribute.}
\label{tab:lookup}
\end{table*}

Let  $\ket{\psi} = \ket{\phi_1} \ldots \ket{\phi_n} \ket{\chi}$ denote the input state to the $(n+k)$-qubit MG circuit, with $\ket{\phi_i}$ being 1-qubit states and $\ket{\chi}$ an arbitrary state on $k$ qubits. As in \cite{Br16} (see also Section \ref{sec:hadamardgadget}), we can construct a matchgate sequence that generates $\ket{\phi_1} \ldots \ket{\phi_n} \ket{\chi}$ from $ \ket{0}^{\otimes n}\ket{+} \ket{\chi}$ and thus reduce the problem to simulating a $(n+k+1)$-qubit MG circuit with input $ \ket{0}^{\otimes n}\ket{+} \ket{\chi}$. Similarly as in \cite{TeDi02}, it is possible to simulate the circuit up to the $i$th intermediate measurement, replace the $i$th intermediate measurement with a projector onto a classically sampled outcome and repeat the procedure in order to simulate up to the $(i+1)$th measurement  (see also Section \ref{sec:simreview}). 

Hence it remains to show that for the input state at hand, Eq.~(\ref{eq:jointprobmany}) may be evaluated within the claimed overall simulation runtime. We first illustrate the evaluation of the mentioned expression for the case of a single adaptive measurement and then generalise the argument to $\mathcal{O}(\operatorname{poly} n)$ adaptive measurements. Let us write the state $\ket{0 \ldots 0}\ket{+}\ket{\chi}$ into the computational basis, $\ket{0 \ldots 0}\ket{+}\ket{\chi} = \sum_{w\in \{0,1\}^{k+1}}  \lambda_w \ket{0 \ldots 0}\ket{w}$. Inserting into Eq.~\eqref{eq:jointprob} leads to
\begin{widetext}
\begin{align}
\label{eq:jointprobcross}
p(x,y_1) &= \sum_{w,w'} \lambda_w \lambda_{w'}^* \sum_{\substack{a_1,b_1,f_1,g_1,\\d_1,e_1, \ldots, d_{|I|},e_{|I|}}} T^1_{j_1,a_1} T^{1*}_{j_1,b_1}  T^{12}_{i_1,d_1}T^{12*}_{i_1,e_1} \ldots T^{12}_{i_{|I|},d_{|I|}} T^{12*}_{i_{|I|},e_{|I|}}  T^1_{j_1,f_1} T^{1*}_{j_1,g_1} \nonumber\\
&\qquad \qquad\qquad\qquad \qquad  \times \bra{\vec{0}}c_{2q_{l'}-1} \ldots c_{2q_1-1} \, A_{a_1 b_1}(y_1)  A_{d_1 e_1, \ldots, d_{|I|} e_{|I|}}(x)  A_{f_1 g_1}(y_1) \,  c_{2p_1-1} \ldots c_{2p_l-1}  \ket{\vec{0}}.
\end{align}
\end{widetext}
Here and in the following we use the same notation as in Section \ref{sec:simreview}.
 Similarly as in \cite{TeDi02}, we have rewritten $\ket{w}$ and $\ket{w'}$ as Majorana operators acting onto $\ket{\vec{0}}$, $\ket{w} = c_{2p_1-1} \ldots c_{2p_l-1}\ket{\vec{0}}$, $\ket{w'} = c_{2q_1-1} \ldots c_{2q_{l'}-1}\ket{\vec{0}}$ (see also Section \ref{sec:simreview}). Note that the $p_i$ and $q_i$ depend on the summation indices $w$ and $w'$. Expressions of the form in Eq.~(\theequation) can be evaluated by individually evaluating each summand of the sum over $w$ and $w'$. 
In \cite{TeDi02}, it has been shown how to efficiently evaluate the summands where $w=w'$  (see also Section \ref{sec:simreview}). The method therein can be straightforwardly  generalised to the cases where $w \neq w'$. Each summand can be rephrased as a Pfaffian of an antisymmetric square matrix $O$ of dimension $l + l' + 2(|I| + 2)$, where $l$ ($l'$) is the Hamming weight of $w$ ($w'$). The matrix $O$ can be easily constructed by going through the pairs $(i,j)$ of Majorana operators ordered as in Eq.~(\theequation) and consulting Table \ref{tab:lookup} in order to read off the matrix entry $O_{i, j}$. In case of $|J|$ adaptive measurements, the matrix $O$ may be similarly constructed and has a dimension of $l + l' + 2(|I| + 2 |J|)$.

Let us now consider the runtime of the algorithm. As the state $\ket{0 \ldots 0}\ket{+}\ket{\chi}$ can be decomposed into at most $2^{k+1}$ computational basis elements, we have to evaluate up to $4 \times  2^{2k}$ $w$-, $w'$-summands within Eq.~(\theequation) and its generalisation for more adaptive measurements. Evaluating each of those summands can be done in $\operatorname{poly}(n)$ time, and moreover, the procedure has to be repeated for each of the $\text{poly}(n)$ many adaptive measurements. Hence, the overall runtime of the simulation algorithm is $\mathcal{O}(\text{poly}(n) 2^{2k})$. Note that a factor $2$ can be saved as matrix elements $\bra{w'}.\ket{w}$ vanish, unless $w$ and $w'$ both have even, or both have odd Hamming weight.
\end{proof}

Clearly, the simulation of MG circuits as in Theorem \ref{theo:fewentangled} remains efficient if up to $k = \mathcal{O}(\operatorname{log}n)$ lines start out in an entangled input state.

\section{Classical simulation Classification of MG computations}
\label{sec:hardnessclassification}

We now complete the classification of classical simulation complexity for MG computations under various conditions on the input state, type of measurements, the number of lines measured and the type of simulation required, as given in Fig. \ref{fig:mg_magic_complexity}.

First, we show that the scenario OUT(1), IN(PROD), ADAPT, STRONG is $\# P$--hard (see Theorem \ref{TH:sharpP}). This implies that in the same setting, but with the more general input IN(MAGIC), simulation is still $\# P$--hard. Similarly the setting with a more general output OUT(MANY) remains $\# P$--hard too. Then we show that the simulation complexity of the NONADAPT case with OUT(MANY) and IN(MAGIC) is also $\# P$--hard (see Theorem \ref{thm:MIMOsharpP}). Finally we prove that considering the previous case but with weak simulation, the existence of an efficient classical simulation would imply that the polynomial hierarchy PH collapses (see Theorem \ref{thm:phcollaps}). We heavily use proof techniques presented in \cite{JoVa14}, which have been used there to prove similar results in the context of Clifford computations, and suitably adapt them to prove the results outlined above.

Let us now show that the scenario OUT(1), IN(PROD), ADAPT, STRONG is $\# P$--hard.
\begin{theorem} \label{TH:sharpP}
Let $\cal A$ be the set of processes defined by adaptive MG circuits with product state inputs, and single bit outputs. Strong classical simulation of $\cal A$ is \#P-hard.
\end{theorem}

\begin{proof}
The strategy employed in Theorem 2 of~\cite{JoVa14} to prove \#P-hardness of strong simulation of adaptive Clifford circuits with computational basis input and single bit outputs is applicable here with several key differences. In what follows, we will emphasise the differences between~\cite{JoVa14} and MG circuits.

Any Boolean function $f$ from $n$ bits to one bit can be implemented using a sequence of Toffoli and $X$ gates together with further ancilla bits initialised to 0. Hence, for input register restricted to being a computational basis state $\ket{x}$ (and omitting any ancillas used), the map ${\cal A}_f: |x\rangle |0\rangle \mapsto |x\rangle|f(x)\rangle$ can be realised by such a gate sequence. Given the fact that one can strongly simulate any circuit in the set $\cal A$, we show that then, one can compute the number of input strings for which $f$ evaluates to 1. In the Clifford case the sequence of Toffoli + $X$ gates can be realised as follows. There, a Toffoli gate can be implemented by measuring one of the control bits and conditionally applying a C-NOT to the second control bit and the target bit. For MG circuits, a similar construction is possible. In order to implement a Toffoli gate on a computational basis state, both of the control lines are measured in the computational basis and a conditional $X$ is implemented on the target line. In MG circuits, it is only possible to apply $X$ on pairs of neighbouring lines, i.e., $X_i \otimes X_{i+1}$, which is easily checked to be a MG. The application of a single $X$ on a line can then be achieved by introducing an auxiliary line $\ket{0}$ and an $\mathcal{O}(n)$ length ladder of $X_j \otimes X_{j+1}$ gates in order to get rid of the second (undesired) $X$ until it acts on the auxiliary line.

A second ingredient of the proof in \cite{JoVa14} is to provide a uniformly random input string. There, this is achieved by starting from $\ket{0}^{\otimes n}$, applying a Hadamard on all lines and then measuring all lines. The Hadamard is not a matchgate, however, as we are allowing product state inputs we can simply directly start with $\ket{+}^{\otimes n}$ and measure all lines. 

The whole process of running ${\cal A}_f$ on a random input $x$ is thus in the set $\cal A$.
Strong simulation is now tantamount to calculating of the number of input strings for which $f$ evaluates to one, $\#f$ viz. measuring the second register after the application of ${\cal A}_f$ yields outcome one with a probability of $\#f /2^n$ whose value is provided by the strong simulation. Hence, strong simulation is \#P-hard, as in \cite{JoVa14}.
\end{proof}

Now, we show that also the task of strongly simulating the NONADAPT case with OUT(MANY) and IN(MAGIC) is $\# P$--hard.

\begin{theorem}
\label{thm:MIMOsharpP}
Let ${\cal A}$ be a set of processes defined by non-adaptive MG circuits with magic state and product state inputs and multiple bit outputs. Strong classical simulation of ${\cal A}$ is \#P-hard.
\end{theorem}

\begin{proof}
We will follow the outline of the proof of Theorem 6 from~\cite{JoVa14}, which is the Clifford counterpart of the present theorem. 
The main idea is to show that if it is possible to efficiently strongly simulate the output of MG circuits associated with ${\cal A}$, then it would also be possible to strongly efficiently simulate universal quantum computations. This, in turn, would allow  calculation of $\#f$ for an arbitrary Boolean function $f$, since $\#f/2^n$ can be realised as an output probability of a suitable quantum computation (evaluation of $f$ on an equal superposition of all its inputs).

It remains to show that strong simulation of universal quantum computation can be reduced to a strong simulation of MG circuits given by $\cal A$.
To this end, consider an arbitrary (universal) quantum circuit $D$ in terms of MGs and SWAP gates. Now, consider a `gadgetised' version of the circuit, $D'$, in which all SWAP gates are replaced by SWAP-gadgets and magic states $\ket{M}$ (see also Section \ref{sec:magic4}) \cite{HeJo19}. Recall that the SWAP-gadget involves adaptive measurements followed by adaptive Pauli corrections and (adaptive) $G(\pm Z, X)$ to swap out auxiliary lines. Instead of this, however, let us omit the intermediate measurements in the gadget as well as the Pauli-corrections. Let us directly use a sequence of $fSWAP = G(+Z, X)$ gates to `swap' the lines in question to the bottom of the circuit. Then, at the end of the computation, these lines are measured in the computational basis. 
In case all the measurements on the auxiliary lines yield outcome 0, the remaining lines of $D'$ are prepared in the output state of $D$. Otherwise, however, this is not necessarily the case.
Let $K$ denote the number of SWAP gates in the circuit and let us denote the additional lines that were added for the $K$ SWAP-gadgets by $a_1, \ldots, a_{4K}$. Then, the output probability of measuring a bit string $y$ on the qubits of interest in the circuit $D$ is given by
\begin{align}
\operatorname{prob}_D(y) &= \operatorname{prob}_{D'}(y | 0_{a_1} \ldots 0_{a_{4K}}) \nonumber \\
				&= \frac{\operatorname{prob}_{D'}(y \ 0_{a_1} \ldots 0_{a_{4K}})}{\operatorname{prob}_{D'}(0_{a_1} \ldots 0_{a_{4K}}) }.
\end{align}
Strong simulability of $D'$ allows the evaluation of both quantities in the quotient in Eq.~(\theequation), and thus strong simulation of $D$. As $D'$ is a computational task as defined by ${\cal A}$, the statement follows.
\end{proof}

For our next result, recall that Observation \ref{obs:entangledinput} asserted that the setting of non-adaptive MG circuits with IN(MAGIC) and OUT(1) can be classically efficiently simulated in the strong, and hence also in the weak sense. However, allowing for OUT(MANY) i.e. a multiple number  $\mathcal{O}(n)$ of outputs, is a setting that is unlikely to be even weakly simulable, as we show in the following. (Theorem \ref{thm:MIMOsharpP} above has shown that for this setting, strong, but not necessarily weak, simulation is \#P-hard.) More precisely we will show that efficient weak simulability of this setting would imply collapse of the polynomial hierarchy PH.

\begin{theorem}
\label{thm:phcollaps}
Let $\cal A$ to be the set of processes defined by non-adaptive MG circuits with magic state and product state inputs, and multi-line outputs. Efficient weak classical simulability of $\cal A$ would imply a collapse of PH to its third level. 
\end{theorem}
\begin{proof}
The technique used to prove this theorem is very similar to that used to prove of its counterpart for Clifford circuits (see Theorem 7 of \cite{JoVa14}):  we will consider a setting in which, additionally, it is possible to post-select measurement outcomes produced by a process defined by $\cal A$. We will argue that within this post-selected variant of  $\cal A$, universal quantum computation with post-selection is possible. Then the proof is completed by the fact  \cite{BrJo11} that efficient weak simulability of any class of quantum circuits, whose post-selected version gives rise to post-selected universal quantum computation, implies collapse of PH. Note that the probability of the measurement outcomes, on which one post-selects, does not play a role as long as it is non-vanishing. 

It remains to show that  $\cal A$ with post-selection indeed gives rise to post-selected universal quantum computation. To this end, consider an arbitrary (universal) quantum circuit in terms of MGs and SWAP gates. As in the proof of Theorem \ref{thm:MIMOsharpP}, consider a `gadgetised' version of the circuit and just as there, let us omit the adaptive measurements, the Pauli-corrections, and use fSWAP gates to dispose the auxiliary lines, which are only measured at the end of the computation. Again, this may potentially corrupt the state of all of the involved lines. To avoid this, the measurements on the auxiliary lines are post-selected on obtaining the outcome 0. It can be readily verified that this procedure has the same effect as actually performing the intermediate measurements and adapting on the outcomes, and hence achieves universal quantum computation. Note that the described process is a post-selected ${\cal A}$ process, and thus completes the proof of the theorem.
\end{proof}

\section{Additional results on the computational power of MG circuits}
\label{sec:additionalSimu}

In this section, we present two additional settings for MG circuits complementing those that have already been considered. First, we consider MG circuits with adaptive measurements in the computational basis and input states in which fewer than four consecutive qubits may be entangled. 
We show that two-qubit entangled input states (on consecutive lines) remain classically efficiently simulable, while entanglements over sets of three consecutive lines allow universal quantum computation. Then, secondly, we consider MG circuits with computational basis inputs and adaptive measurements in {\em arbitrary} (1-qubit) bases. While measurements in the computational basis have been shown to be classically efficiently simulable, here we show that if adaptive measurements in any additional basis are available, then universal quantum computation becomes possible.

\subsection{Few-qubit entangled input states on consecutive lines}\label{sec:twoqubitinput}

As mentioned previously, the setting of MG circuits with product input states, final 1-qubit measurements in arbitrary bases and adaptive measurements in the computational basis, is weakly classically efficiently simulable~\cite{Br16}. In contrast to this (cf our discussion of magic states for MG circuits), if in this setting we allow 4-qubit entangled states on consecutive lines in the input then universal quantum computation becomes possible~\cite{HeJo19}. Hence the question of the simulation complexity of two- as well as three-qubit entangled input states naturally arises: how does the  computational power transition from the classically simulable single-qubit input setting into universal quantum computation for four- or more-qubit entangled input states?

In the following we show that MG circuits in the above setting with two-qubit entangled input states on consecutive lines remains classically simulable. To this end, we will use  the Hadamard gadget technique of \cite{Br16} (cf also Section \ref{sec:hadamardgadget}). Then we will show that the 3-qubit setting allows universal quantum computation.

\begin{theorem}
\label{theo:twoqubitinput}
Consider the setting of  $n$-qubit MG circuits with final single-qubit measurements in arbitrary bases, and adaptive measurements in the computational basis. Let the input state be an $n$-qubit state $\ket{\psi}$ that is a tensor product of arbitrary two-qubit entangled states on neighboring lines $\ket{\psi} = \ket{\alpha}_{12}\ket{\beta}_{34} \cdots \ket{\delta}_{(n-1) n}$. Then the resulting processes are classically weakly efficiently simulable.
\end{theorem}

\begin{proof}
As explained in Section \ref{sec:hadamardgadget}, a MG circuit can be constructed that transforms a state $\ket{0}^{\otimes n}\ket{+}$ into an arbitrary $n$-qubit product state \cite{Br16}. It has been proven in \cite{Br16} that an adaptive MG circuit with the input state $\ket{0}^{\otimes n}\ket{+}$ is classically weakly efficiently simulable. Here we show that there exists an adaptive MG circuit that transforms the input state $\ket{0}^{\otimes n+\lfloor n/2 \rfloor}\ket{+}$ to arbitrary 2-qubit entangled states on neighbouring lines, $\ket{\psi} = \ket{\alpha}_{12}\ket{\beta}_{34} \cdots \ket{\delta}_{(n-1) n}$. With the above result of \cite{Br16}, this will complete the proof of Theorem \ref{theo:twoqubitinput}.
 
To achieve the above claimed reduction, we construct a small MG circuit involving adaptive measurements in the computational basis which, when applied to the $(n+\lfloor n/2 \rfloor +1)$-qubit product state $\ket{\psi_0} = \ket{+}\ket{00}\ket{+}\ket{00}\ket{+}\ldots\ket{+}\ket{00}\ket{+}$, generates arbitrary two-qubit states on neighbouring lines.
Considering $\ket{\psi_0}$, let us call those lines that are initialised in $\ket{+}$ auxiliary lines and let us call those lines that are initialised in $\ket{0}$ computational lines. The computational lines will eventually carry the two-qubit entangled states while the auxiliary lines will be measured in the process.

The Hadamard gadget introduced in \cite{Br16} (see also Section~\ref{sec:hadamardgadget}) makes it possible to implement Hadamard gates on a line juxtaposed to any line whose state is $\ket{+}$, i.e., on all the computational lines. Recall that single qubit phase gates are matchgates. Then, with availability of the Hadamard gate, one can implement arbitrary single-qubit unitaries on the computational lines. Moreover, recall that gates of the form $e^{i (\alpha X \otimes X + \beta Y \otimes Y)}$, where $\alpha, \beta \in \mathbb{R}$, are MGs and may thus be implemented on (neighbouring) computational line pairs (cf Section \ref{sec:simreview} and also \cite{BrGa11}). This suffices to implement arbitrary two-qubit unitaries on the computational line pairs, which may be seen e.g. using the decomposition of two-qubit unitaries given in \cite{KrCi01}. See also Figure \ref{fig:twoqubitreduction} for a picture of the construction. Finally, the $n/2$ auxiliary lines are discarded by first measuring them in the computational basis and then depending adaptively on the measurement outcome,  swapping them out using either $fSWAP$ or $G(-Z,X)$. One thus ends up with $n$ consecutive lines initialised in any desired two-qubit states on consecutive lines, completing the proof of the Theorem.
\end{proof}

\begin{figure*}[t]
\includegraphics[width=0.8\linewidth]{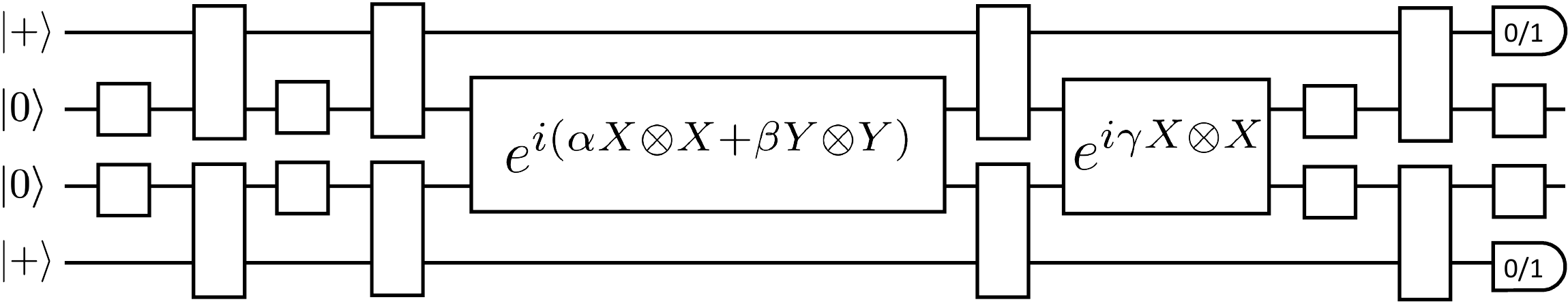}
\caption{Sketch of the gadget described in the proof of Theorem \ref{theo:twoqubitinput}. The circuit gadget allows the preparation of an arbitrary (not only fermionic) two-qubit state with the help of auxiliary lines in the state $\ket{+}$. To this end, the Hadamard gadget introduced in \cite{Br16} is utilised, as well as a decomposition of arbitrary two-qubit unitaries from \cite{KrCi01} and the Euler decomposition of single-qubit unitaries. Here, blank single-qubit operations denote phase gates and blank two-qubit operations denote the Hadamard gadget. Finally, the auxiliary lines in the state $\ket{+}$ are discarded as described in the proof.}
\label{fig:twoqubitreduction}
\end{figure*}

Let us remark here that it is crucial to have the entangled two-qubit input states being on {\em neighbouring} lines in order to have a classically efficiently simulable situation. Indeed, the setting of entangled two-qubit input states on non-neighbouring lines, such as e.g. the magic state $\ket{M}$, actually allows universal quantum computation. In~\cite{Br16} it is asserted that the overall ordering of the qubits in MG circuits is irrelevant for issues of efficient classical simulability. However the argument given there (involving translating the MG process back to a process on fermionic modes, then relabelling the modes and finally translating back to a MG description), does not apply to entangled input states (e.g. generally altering the range of entanglements) so there is no contradiction with our results. We discuss this issue further in Appendix \ref{sec:reordering}.

Finally here, we address three-qubit entangled input states. It can be easily seen that the circuit gadget constructed in the proof of Theorem \ref{theo:twoqubitinput} may be generalized to generate certain sub-families of three-qubit states. However, such constructions seem to not allow the generation of all three-qubit states. Indeed note that e.g. two copies of the three-qubit GHZ-state on neighbouring lines may be converted into a copy of the magic state $\ket{M}$ via matchgates and adaptive measurements in the computational basis (cf also \cite{HeJo19}). Thus, MG circuits with adaptive measurements in the computational basis supplied with certain three-qubit entangled input states (each on consecutive lines) can allow universal quantum computation.

Altogether, we thus obtain the following picture of the transition from a classically simulable situation into universal quantum computation for few-qubit entangled input states. We have that the classical simulability of single-qubit input states is still retained for entangled two-qubit input states on consecutive lines. This is not the case any more for entangled three-qubit states on neighboring lines. While certain subclasses of three-qubit states remain classically simulable, others, however, may be converted into magic states and thus allow universal quantum computation. Then, considering four- or more-qubit states, clearly universal quantum computation is possible, as fermionic non-Gaussian states (and thus magic states) exist within these states. Furthermore, if 2-qubit entangled input states on non-consecutive lines are allowed, this recovers universal quantum computation too.

\subsection{MG circuits with adaptive measurements in arbitrary bases}
\label{sec:adaptivemeas}

In this subsection we study the computational power of MG circuits with product state input and adaptive measurements in arbitrary bases. While it has been show that adaptive measurements in the computational basis are classically efficiently simulable \cite{Br16}, not much is known about adaptive measurements in differing bases. First, we start out with a simple scenario, in which a specific additional measurement basis is available, namely measurements in the basis $\{\ket{+},\ket{-}\}$. We observe that the possibility of measuring in this basis (in addition to the computational basis) allows universal quantum computation. Then we show that more generally, the possibility of adaptively measuring in the computational basis as well as any additional basis (sufficiently distinct from the computational basis) allows universal quantum computation.

\begin{observation}
\label{obs:adaptplus}
MG circuits with computational basis input and adaptive measurements in the computational basis as well as the $\{\ket{+},\ket{-}\}$-basis allow universal quantum computation.
\end{observation}
\begin{proof}
We show that with the described adaptive measurements and a single auxiliary line at hand, Hadamard gates may be implemented at any place within a MG circuit at the cost of two adaptive measurements and $\mathcal{O}(n)$ MGs per Hadamard. Then the observation  follows from the fact that MGs together with the Hadamard form a universal gate set\footnote{This can be seen by recalling that phase gates (which are matchgates), together with the Hadamard gate allow implementation of an arbitrary single qubit gate, and that these with matchgates allow implementation of arbitrary two-qubit gates.}.

The construction requires one auxiliary line (e.g. at the bottom of the circuit) which is initialised in either $\ket{0}$ or $\ket{1}$. Suppose one wishes to implement a Hadamard gate at a certain point in the circuit, on a specified line (target line). Then, the first step is to bring the auxiliary line into a position neighbouring to the target line (e.g. by using a sequence of  $\mathcal{O}(n)$ $fSWAP$s or $G(-Z,X)$s). Then the auxiliary line is measured in the basis $\{\ket{+},\ket{-}\}$. Adaptively applying $Z$ then deterministically prepares the state $\ket{+}$. Now, the Hadamard gadget~\cite{Br16} is utilised to implement the Hadamard on the target line. Finally the auxiliary line is removed as follows:  a computational basis measurement is performed on it and once more $fSWAP$s or $G(-Z,X)$ are used to return the auxiliary qubit to the position from which it was swapped in initially. It may be then reused to implement another Hadamard  later on.
\end{proof}

Let us remark here that, as discussed in \cite{HeJo19}, in view of the above proof one might spuriously think that a MG circuit with adaptive measurements and input states containing $\ket{+}$ might also allow universal quantum computation. Opposed to that, however, the result in \cite{Br16} shows that such a situation is classically efficiently simulable. An additional key ingredient, which adaptive measurements in the basis $\{\ket{+},\ket{-}\}$ bring in, is the possibility to obtain $\ket{+}$ at any place needed, which is not possible when they are merely supplied as input state, as SWAP is not an allowed operation and fSWAP would entangle the auxiliary line with the other lines. And if $\ket{+}$'s were initially placed amongst the input qubits at all needed positions, they would fragment the remaining input into sectors for n.n.  MG actions until they are used.

Observation \ref{obs:adaptplus} straightforwardly generalises to measurements in any basis of the form $\{ \ket{0} + e^{i \phi} \ket{1}, \ket{0} - e^{i \phi}\ket{1}\}$ for $\phi \in \mathbb{R}$. This is due to the fact that single qubit phase gates are matchgates and may thus be utilised prior to an adaptive measurement in order to effectively implement a measurement in the basis $\{\ket{+},\ket{-}\}$ instead.
In the following theorem, we generalise Observation \ref{obs:adaptplus} to adaptive measurements in any basis differing from the computational basis. As phase gates are MGs we will without loss of generality consider measurements in the bases $\{\cos x\, \ket{0} + \sin x\, \ket{1},\sin x\, \ket{0} - \cos x\, \ket{1}\}$, where $x \in (0, \pi/4]$, in the following.

\begin{theorem}
\label{theo:adaptarbitrary}
MG circuits with computational basis input and adaptive measurements in the computational basis as well as the $\{\cos x\, \ket{0} + \sin x\, \ket{1},\sin x\, \ket{0} - \cos x\, \ket{1}\}$-basis for some $x \in (0, \pi/4]$ allow universal quantum computation. Here, $x$ may be either constant (independent of $n$), or lower bounded away from 0 by some inverse polynomial in $n$.
\end{theorem}

\begin{proof}
We will show that by means of the adaptive measurements in the two available bases, it is possible to construct the state $\ket{+}$ with high probability at arbitrary positions within the circuit. Then as in the proof of Observation \ref{obs:adaptplus}, universal quantum computational power is obtained.

Let us now describe a procedure creating the state $\ket{+}$. First, similarly as in the proof of Observation \ref{obs:adaptplus}, two (rather than one, as there) auxiliary lines in the state $\ket{0}$ are swapped into the desired position. Then, both of the auxiliary lines are measured in the second (the non-computational)  basis. Adaptively performing $Z$ corrections, one deterministically obtains $\cos x\, \ket{0} + \sin x\, \ket{1}$ on both lines. Next, the MG $G(A,B)$ is applied, where $A =  \begin{pmatrix} \sin y & \cos y \\  \cos y & -\sin y \end{pmatrix}$ and  $B = \begin{pmatrix} \sin z & \cos z \\  \cos z & -\sin z \end{pmatrix}$, where $y,z \in \mathbb{R}$ will be specified shortly. A computational basis measurement is performed on the first auxiliary qubit. It can be easily verified that in case of obtaining the measurement outcome 0, the state of the second auxiliary qubit is proportional to $[\sin^2x \cos y + \cos^2x \sin y]\, \ket{0} + [\sin x \cos x (\sin z + \cos z )]\,\ket{1}$. Thus, choosing $z=\pi/4$ (i.e. $B$ is a Hadamard operation) and $y = \arctan\frac{\tan x (\sqrt{2}  + \tan x)}{1 + \sqrt{2} \tan x}$, the second qubit is in the state $\ket{+}$. Moreover, the probability of obtaining measurement outcome 0 is given by $p(0) = \sin^2(2x)$. Here $y$ and $z$ have been chosen to maximise  the probability of obtaining $\left|+\right>$ within the described construction. See also Figure \ref{fig:adaptgadget} for a sketch of the described sequence. In case one instead obtains the undesired measurement outcome 1, the second auxiliary qubit will also be in the state $\ket{1}$ and hence the process must be iterated.

\begin{figure}[ht]
\includegraphics[width=0.8\linewidth]{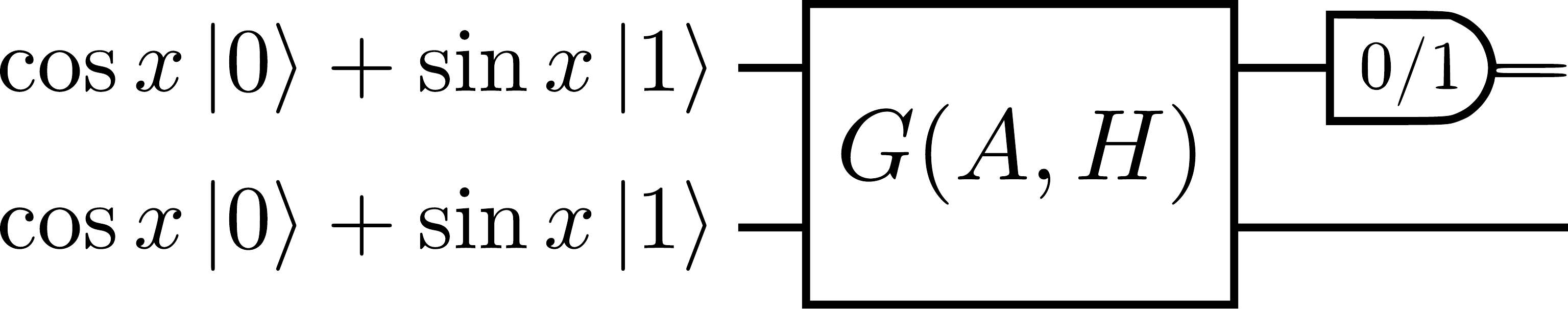}
\caption{Sketch of a circuit gadget creating a copy of $\ket{+}$ with probability $\sin^2(2x)$. Here $A$ is as given in the text. In case of success (measuring outcome 0 on the first line), the second line carries the state  $\ket{+}$.}
\label{fig:adaptgadget}
\end{figure}
If $x$ is a constant which is strictly larger than 0, or $x$ is bounded away from 0 by an inverse polynomial in $n$, an acceptably small number of repetitions of the described process suffice to generate $\ket{+}$ with sufficiently high probability, exponentially close to 1. More precisely, $\ket{+}$ may be obtained in the two cases with an arbitrarily high probability $1 - \epsilon$ using $\mathcal{O}(\operatorname{poly}(1/\epsilon))$  resp. $\mathcal{O}(\operatorname{poly}(n,1/\epsilon))$ adaptive measurements. This completes the proof.
\end{proof}

\section{Discussion}

In this work we have studied the complexity of classically simulating MG circuits for a wide variety of settings. We have compared our results for MGs with their counterparts for Clifford circuits, although these two gate sets are of a very different nature. MG circuits are subject to a n.n. locality constraint, while Clifford circuits are not. And Clifford gates form a discrete gate set, while the MGs form a continuous set. Yet, notably, despite these differences we observed that with respect to complexity of simulation, their behaviour showed some striking parallels.

Clifford computations have been very much studied and naturally, the question arises whether there may be fruitful MG analogues of further Clifford computation constructions. For example, could it be possible to devise hybrid quantum-classical computation schemes based on MG circuits, akin to those that have been developed for Clifford circuits in \cite{BrSm16,BrGo16}? Moreover, it would be interesting to see whether the classical simulation results obtained here for MGs may be used for issues of verification of quantum computations expressed in terms of extended MG circuits, akin to the ideas for Clifford circuits in \cite{JoSt17}. 

Another question which remains open is how the scaling of classical simulation methods of simulating MG circuits with adaptive measurements and few magic states, may be improved compared to the plain scaling result in Theorem \ref{theo:fewentangled}.
 The analogous question in the Clifford setting has been the subject of intense study~\cite{BrSm16}. The main tool there was the introduction of the notion of stabiliser rank of a state $\ket{\psi}$ viz. the number of components in the smallest decomposition of $\ket{\psi}$ in terms of (weighted) stabilizer states. Finding the smallest decomposition for a state $\ket{\psi}$ which is given by a tensor product of magic states has a direct impact on the efficiency of direct classical simulation methods for the universal Clifford + $T$ gate set. One can introduce an analogous quantity for MG circuits: given an $n$-qubit state $\ket{\psi} = \sum_{i=1}^{\chi}\alpha_i\ket{\gamma_i}$ its {\it Gaussian rank} of $\ket\psi$ is the smallest $\chi$ for which it admits the decomposition where the $\ket{\gamma_i}$ are Gaussian states, i.e, states that may be generated by MG circuits from computational basis states. Our magic state $\ket{M}$ for MG circuits has a Gaussian rank of 2. However, it is unclear whether multiple ($k$) copies of it have a Gaussian rank that is smaller than $2^k$, or indeed how a small Gaussian rank may be utilised to advantage in classical simulation methods for MGs (see proof of Theorem \ref{theo:fewentangled}).

\begin{acknowledgments}
We thank Mithuna Yoganathan for helpful discussions and suggestions on the content of this work, especially in relation to Figure 1.
 M.H. and B.K. acknowledge financial support from the Austrian Science Fund (FWF) grant DK-ALM: W1259-N27 and the SFB BeyondC. Furthermore, B.K. acknowledges support of the Austrian Academy of Sciences via the Innovation Fund ``Research, Science and Society''. R.J. and S.S. acknowledge support from the QuantERA ERA-NET Cofund in Quantum Technologies implemented within the European Union's Horizon 2020 Programme (QuantAlgo project), and administered through the EPSRC grant EP/R043957/1., and S.S. by the Leverhulme Early Career Fellowship scheme.
\end{acknowledgments}

\appendix

\section{Remarks on the reordering of lines in matchgate circuits}
\label{sec:reordering}
In~\cite{Br16} it is asserted that the overall ordering of the qubits in MG circuits is irrelevant for issues of efficient classical simulability (cf paragraph before the one containing eq. (7) in~\cite{Br16}). This applies straightforwardly if the reordering of the input state and the final measurements is possible via fermionic SWAPs as applies to computational basis inputs and final measurements in the computational basis. In this setting, applying fSWAPs effectively corresponds to reordering the lines. 
Furthermore the method of \cite{Br16}, involving relabelling of fermionic modes after reinterpreting the MG circuit as a free fermionic evolution, also extends the use of fSWAPS to be applicable to the reordering of general product state inputs too. However, importantly, as mentioned in the main text after our proof of Theorem \ref{theo:twoqubitinput}, in a setting involving entangled state inputs, line reordering can have a dramatic effect on classical simulation complexity.

We mention also that as long as no adaptive measurements are performed, a different argument can be used to address reordering issues in the case of product state input:  it is shown in \cite{JoMi08} that the technique of classical simulation there can be employed for product state input (and for few entangled input states as discussed in Section \ref{sec:clsim}) even if the whole MG circuit is conjugated by an arbitrary Clifford operation (which includes permuting the lines). This then implies that reordering in such a setting does not alter the simulation complexity.

Let us finally briefly mention the setting of MG circuits with product state inputs and $k$ final measurements in arbitrary bases, which has been shown in \cite{Br16} to be weakly classically efficiently simulable.  The Hadamard gadget is employed there not only to prepare a product state at the beginning of the circuit (cf Section \ref{sec:hadamardgadget}), but also to reduce the final general basis measurements to computational basis measurements. The construction initially applies only to the measurements being on the final $k$ lines and is then generalised to any subset of $k$ lines by invoking the reordering method of \cite{Br16} mentioned above. We note here an alternative way of achieving this final step. The case of $k = n$ could be initially considered and sampled by weak classical simulation (being an example of final contiguous lines) and then a sample from the probability distribution corresponding to measuring any $k<n$ of the lines may be obtained by disregarding the outcomes for lines one is not interested in.

\end{document}